\documentclass[envcountsame]{llncs}
\usepackage{makeidx}  
\usepackage{graphicx}
\usepackage{xspace}
\usepackage{url}
\usepackage{amssymb}
\usepackage{latexsym}
\usepackage{amsmath}
\usepackage{amstext}
\usepackage{times}
\usepackage{amstext}
\usepackage{calc}
\usepackage{amsopn}
\usepackage[noend]{algorithmic}
\usepackage[boxed]{algorithm}
\usepackage{eucal}
\usepackage{latexsym}
\usepackage{tabularx}
\usepackage{subfig}

\usepackage{tikz}
\usetikzlibrary{snakes}

{\bf}{\rm}

\newcommand{\Z}{\ensuremath{\mathbb{Z}}}
\newcommand{\N}{\ensuremath{\mathbb{N}}}

\newcommand{\pw}{\ensuremath{\mathbf{pw}}}

\newcommand{\bag}[1]{B_{#1}}
\newcommand{\subgraph}[1]{G_{#1}}
\newcommand{\subbags}[1]{V_{#1}}
\newcommand{\subedges}[1]{E_{#1}}

\newcommand{\targetW}{W}

\newcommand{\iterE}{\overline{\sigma}}
\newcommand{\iterV}{\overline{\tau}}
\newcommand{\itere}{\sigma}
\newcommand{\iterv}{\tau}
\newcommand{\itercE}{\overline{\mu}}
\newcommand{\itercV}{\overline{\nu}}
\newcommand{\iterce}{\mu}
\newcommand{\itercv}{\nu}

\newcommand{\eiterE}{\overline{\mathbf{y}}}
\newcommand{\eiterV}{\overline{\mathbf{x}}}
\newcommand{\eitere}{\mathbf{y}}
\newcommand{\eiterv}{\mathbf{x}}
\newcommand{\eitercE}{\overline{\mathbf{cy}}}
\newcommand{\eitercV}{\overline{\mathbf{cx}}}
\newcommand{\eiterce}{\mathbf{cy}}
\newcommand{\eitercv}{\mathbf{cx}}

\newcommand{\treedecomp}{\mathbb{T}}

\newcommand{\vX}{\overline{X}}
\newcommand{\vY}{\overline{Y}}
\newcommand{\eX}{X}
\newcommand{\eY}{Y}
\newcommand{\vFX}{\overline{FX}}
\newcommand{\vFY}{\overline{FY}}
\newcommand{\eFX}{FX}
\newcommand{\eFY}{FY}

\newcommand{\Xmarkers}{\overline{\mathcal{M}^X}}
\newcommand{\Ymarkers}{\overline{\mathcal{M}^Y}}
\newcommand{\xmarkers}{\mathcal{M}^X}
\newcommand{\ymarkers}{\mathcal{M}^Y}
\newcommand{\sols}{\mathcal{S}}
\newcommand{\cand}{\mathcal{R}}
\newcommand{\objs}{\mathcal{C}}

\newcommand{\compnomark}{\mathtt{cc}}

\floatname{algorithm}{Algorithm}

\renewcommand{\subset}{\subseteq}

\newcommand{\defproblemu}[3]{
  \vspace{1mm}
\noindent\fbox{
  \begin{minipage}{\textwidth}
  #1 \\
  {\bf{Input:}} #2  \\
  {\bf{Question:}} #3
  \end{minipage}
  }
  \vspace{1mm}
}

\newcommand{\zero}{{\mathbf{0}}}

\newcommand{\oneone}{{\mathbf{1}_1}}
\newcommand{\onetwo}{{\mathbf{1}_2}}

\newcommand{\ff}{\varphi}
\newcommand{\dia}{\diamondsuit}
\newcommand{\rect}{\square}
\newcommand{\Ii}{\mathcal{I}}
\newcommand{\Hh}{\mathcal{H}}
\newcommand{\Pp}{\mathcal{P}}
\newcommand{\Xx}{\mathcal{X}}
\newcommand{\Yy}{\mathcal{Y}}
\newcommand{\Gg}{\mathcal{G}}

\newcommand{\hampath}{{\sc Hamiltonian Path}\xspace}

\newcommand{\longestpath}{{\sc Longest Path}\xspace}

\newcommand{\longestcycle}{{\sc Longest Cycle}\xspace}
\newcommand{\gmtsp}{{\sc Graph Metric Travelling Salesman Problem}\xspace}

\newcommand{\krecognition}{$\mathcal{K}$-{\sc Recognition}\xspace}
\newcommand{\rdomset}{$r$-{\sc Dominating Set}\xspace}
\newcommand{\vertexcover}{{\sc Vertex Cover}\xspace}
\newcommand{\cvertexcover}{{\sc Connected Vertex Cover}\xspace}

\newcommand{\cdomset}{{\sc Connected Dominating Set}\xspace}

\newcommand{\steinertree}{{\sc Steiner Tree}\xspace}

\newcommand{\fvs}{{\sc Feedback Vertex Set}\xspace}

\newcommand{\cfvs}{{\sc Connected Feedback Vertex Set}\xspace}

\newcommand{\coct}{{\sc Connected Odd Cycle Transversal}\xspace}

\newcommand{\mincyclecovername}{{\sc{Min Cycle Cover}}\xspace}

\newcommand{\exactleaf}{{\sc{Exact $k$-leaf Spanning Tree}}\xspace}

\newcommand{\exactoutbranching}{{\sc{Exact $k$-Leaf Outbranching}}\xspace}
\newcommand{\maxspantree}{{\sc{Maximum Full Degree Spanning Tree}}\xspace}

\newcommand{\disjointpaths}{{\sc{Disjoint Paths}}\xspace}

\newcommand{\lvertexdeletion}{$C_l$-{\sc{Vertex Deletion}}\xspace}
\newcommand{\lcliquevertexdeletion}{$K_l$-{\sc{Vertex Deletion}}\xspace}
\newcommand{\ltlvertexdeletion}{{\sc{Girth $>l$ Vertex Deletion}}\xspace}

\newcommand{\cnfsat}{{\sc{CNF-SAT}}\xspace}
\newcommand{\tcnfsat}{{\sc{3CNF-SAT}}\xspace}

\newcommand{\tcnf}{{\sc{3CNF}}\xspace}
\newcommand{\lagl}{{\sc{ECML}}\xspace}
\newcommand{\laglc}{{\sc{ECML+C}}\xspace}
\newcommand{\laglfull}{{\sc{Existential Counting Modal Logic}}\xspace}
\newcommand{\cmgl}{{\sc{CML}}\xspace}
\newcommand{\cmglfull}{{\sc{Counting Modal Logic}}\xspace}

\begin{document}

  \date{}

  \author{
	      Micha\l{} Pilipczuk
  }
  \institute{
    Faculty of Mathematics, Informatics and~Mechanics, \\
    University of Warsaw, Poland\\
    \email{michal.pilipczuk@students.mimuw.edu.pl}
  }

  \title{Problems parameterized by treewidth tractable in single exponential time: a~logical approach}

\maketitle
\begin{abstract}
We introduce a~variant of modal logic, dubbed \laglfull (\lagl), which captures a~vast majority of problems known to be tractable in single exponential time when parameterized by treewidth. It appears that all these results can be subsumed by the theorem that model checking of \lagl admits an~algorithm with such complexity. We extend \lagl by adding connectivity requirements and, using the Cut\&Count technique introduced by Cygan et al.~\cite{my}, prove that problems expressible in the extension are also tractable in single exponential time when parameterized by treewidth; however, using randomization. The need for navigationality of the introduced logic is justified by a~negative result that two expository problems involving non-acyclic conditions, \lvertexdeletion and~\ltlvertexdeletion for $l\geq 5$, do not admit such a~robust algorithm unless Exponential Time Hypothesis fails.
\end{abstract}

\section{Introduction}

The notion of treewidth, introduced by Robertson and~Seymour in their proof of Wagner's Conjecture \cite{rs:minors3}, in recent years proved to be an~excellent tool for capturing characteristics of certain graph classes. Of particular interest are algorithmic applications of treewidth. Many problems, while hard in general, become robustly tractable, when the input graph is of bounded treewidth --- a~usual technique bases on constructing a~dynamic programming algorithm on the tree decomposition. When combined with the graph-theoretical properties of treewidth, the approach leads to a~number of surprisingly efficient algorithms, including approximation~\cite{bidimensionality,eppstein}, parameterized~\cite{subexponential-bound-genus,enumerate-and-expand} and~exact algorithms~\cite{Fomin06ontwo,rooij:inclusion/exclusion}. In most cases, the dynamic program serves as a~subroutine that solves the problem, when the treewidth turns out to be small.

The tractability of problems parameterized by treewidth can be generalized into a~meta-theorem of Courcelle~\cite{Courcelle}: there exists an~algorithm that, given a~MSO formula $\ff$ and~a~graph $G$ of treewidth $t$, tests whether $\ff$ is true in $G$ in time $f(|\ff|,t)|G|$ for some function $f$. Courcelle's Theorem can be viewed as a~generalization of Thatcher and~Wright Theorem about equivalence of MSO on finite trees and~tree automata; in fact, in the proof one constructs an~analogous tree automaton working on the tree decomposition. Unfortunately, similarly to other theorems regarding MSO and~automata equivalence, the function $f$, which is in fact the time needed to process automaton's production, can depend very badly on $|\ff|$ and~$t$~\cite{revisited,weyer}. Therefore, a~lot of effort has been invested in actual construction of the dynamic programming algorithms mimicking the behaviour of a~minimal bottom-up automaton in order to obtain solutions that can be considered efficient and~further used as robust subroutines. One approach, due to Arnborg et al.~\cite{maximisation}, is extending MSO by maximisation or~minimisation properties, which corresponds to augmenting the automaton with additional counters. In many cases, the length of the formula defining the problem can be reduced to constant size, yielding a~$f(t)|G|^{O(1)}$ time algorithm. Unfortunately, careful analysis of the algorithm shows that the obtained function $f$ can be still disastrous; however, for many concrete problems the algorithm can be designed explicitly and~the complexity turns out to be satisfactory. For example, for the expository \vertexcover problem, the book by Kleinberg and~Tardos gives an~algorithm with running time $4^t|G|^{O(1)}$~\cite{kleinberg-tardos}, while the book by Niedermeier contains a~solution with complexity $2^t|G|^{O(1)}$~\cite{niedermeier:book}.

Recently, Lokshtanov et al.~\cite{treewidth-lower} initiated a~deeper study of currently best dynamic programming routines working in single exponential time in terms of treewidth. For a~number of problems they proved them to be probably optimal: a~faster solution would yield a~better algorithm for \cnfsat than exhaustive search. One can ask whether the phenomenon is more general: the straightforward dynamic programming solution reflecting the seemingly minimal automaton is optimal under believed assumptions. This question was stated by the same set of authors in~\cite{marx:superexp} for a~number of problems based on connectivity requirements, like \cvertexcover{} or~\hampath{}. For these, the considered routines work in time $2^{O(t\log t)}|G|^{O(1)}$, and~a~matching lower bound for one such problem, \disjointpaths, was already established~\cite{marx:superexp}.

Surprisingly, the answer turned out to be negative. Very recently, Cygan et al.~\cite{my} introduced a~technique called Cut\&Count that yields single exponential in terms of treewidth Monte-Carlo algorithms for a~number of connectivity problems, thus breaking the expected limit imposed by the size of the automaton. The results also include several intriguing lower bounds: while problems that include minimization of the number of connected components of the solution are tractable in single exponential time in terms of treewidth, similar tractability results for maximization problems would contradict {\emph{Exponential Time Hypothesis}}. Recall that {\emph{Exponential Time Hypothesis}} (ETH) states that the infinimum of such $c$ that there exists a~$c^n$ algorithm solving \tcnfsat ($n$ is the number of variables), is greater than $1$.

A natural question arises: what properties make a~problem tractable in single exponential time in terms of treewidth? Can we obtain a~logical characterization, similar to Courcelle's Theorem?

\noindent{\bf{Our contribution.}} We introduce a~model of logic, dubbed \laglfull (\lagl), which captures nearly all the problems known to admit an~algorithm running in single exponential time in terms of treewidth. The model consists of a~variation of modal logic, encapsulated in a~framework for formulating computational problems. We prove that model checking of \lagl formulas is tractable in single exponential time, when parameterized by treewidth. In addition to solving the decision problem, the algorithm can actually count the number of solutions. The result generalizes a~number of explicit dynamic programming routines (for example~\cite{jochen,fomin,fiorini,grohe:book,johan-tw}), however yielding significantly worse constants in the bases of exponents.

Furthermore, we extend the \lagl by connectivity requirements in order to show that the tractability result for \lagl can be combined with the Cut\&Count technique of Cygan et al. Again, we are able to show similar tractability for all the problems considered in~\cite{my}, however with significantly worse constants in the bases of exponents.

Finally, we argue that the introduced logic has to be in some sense navigational or~acyclic, by showing intractability in time $2^{o(p^2)}|G|^{O(1)}$ under ETH of two model non-acyclic problem, \lvertexdeletion and~\ltlvertexdeletion for $l\geq 5$, where $p$ is the width of a~given path decomposition.

\noindent{\bf{Outline.}} In Section~\ref{sec:preliminaries}, we introduce the notation and~recall the well-known definitions. We try to follow the notation from~\cite{my} whenever it is possible. In Section~\ref{sec:logic}, we introduce the model of logic. Section~\ref{sec:tractability} contains the main tractability result, while Section~\ref{sec:connectivity} treats of combining it with the Cut\&Count technique. The details of the dynamic program described in Section~\ref{sec:tractability} can be found in Appendix~\ref{sec:tractability-details} and~the proof of the tractability result for the connectivity extension (dubbed \laglc) can be found in Appendix~\ref{sec:con-proof}. In Appendix \ref{sec:formulae}, we present \laglc formulas for all the connectivity problems considered in~\cite{my}. The reader can treat this part as a~good source of examples of formulas of the introduced logic. In Section~\ref{sec:negative}, we prove the intractability results under ETH. Again, the details of the presented reduction can be found in Appendix~\ref{sec:negative-details}. Section~\ref{sec:conclusions} is devoted to concluding remarks and~suggestions on the further study.

\section{Preliminaries and~notation}\label{sec:preliminaries}

\subsection{Notation}\label{ssec:notation}
\label{subsec:not}
Let $G = (V,E)$ be a~(directed) graph. By $V(G)$ and~$E(G)$ we denote the sets of vertices and~edges (arcs) of $G$,
respectively. Let $|G|=|V(G)|+|E(G)|$. For a~vertex set $X \subset V(G)$ by $G[X]$ we denote the subgraph induced by $X$. For an~edge set $X \subset E$, by $V(X)$ denote the set of the endpoints of the edges from $X$, and~by $G[X]$ --- the subgraph $(V(X), X)$.
Note that for an~edge set $X$, $V(G[X])$ may differ from $V(G)$.

In a~directed graph $G$ by connected components we mean the 
connected components of the underlying undirected graph.
For a~subset of vertices or~edges $X$ of $G$, we denote by $\compnomark(X)$ the number of connected components of $G[X]$.

A monoid is a~semigroup with identity. The identity of a~monoid $M$ is denoted by $e_M$, while the operations in monoids are denoted by $+$. We treat the natural numbers $\N$ (nonnegative integers) also as a~monoid with operation $+$ and~identity $0$.

\subsection{Treewidth and~pathwidth}\label{ssec:treewidth}

\begin{definition}[Tree Decomposition, \cite{rs:minors3}]
A \emph{tree decomposition} of a~(undirected or~directed) graph~$G$ is a~tree~$\treedecomp$ in which each 
vertex~$x \in \treedecomp$ has an~assigned set of vertices~$B_x \subseteq V$ 
(called a~\emph{bag}) such that $\bigcup_{x \in \treedecomp} B_x = V$ with the 
following properties:
\begin{itemize}
\item for any $uv \in E$, there exists an~$x \in \treedecomp$ such that 
$u,v \in B_x$.
\item if $v \in B_x$ and~$v \in B_y$, then $v \in B_z$ for all $z$ on 
the path from $x$ to $y$ in $\treedecomp$.
\end{itemize}
\end{definition}

The \emph{treewidth}~$tw(\treedecomp)$ of a~tree decomposition~$\treedecomp$ is the size of 
the largest bag of $\treedecomp$ minus one. The treewidth of a~graph $G$ is the 
minimum treewidth over all possible tree decompositions of~$G$. A~{\em path decomposition} is a~tree decomposition that is a~path.
The pathwidth of a~graph is the minimum width over all path decompositions.

We use a~modified version of tree decomposition from~\cite{my}, called {\emph{nice tree decomposition}}, which is more suitable for development of dynamic programs. The idea of adjusting the tree decomposition to algorithmic needs comes from Kloks~\cite{Kloks94}.

\begin{definition}[Nice Tree Decomposition, Definition 2.3 of~\cite{my}]
A \emph{nice tree decomposition} is a~tree decomposition with one special 
bag $r$ called the \emph{root} with $B_r = \emptyset$ and~in which each 
bag is one of the following types:
\begin{itemize}
\item \textbf{Leaf bag}: a~leaf $x$ of $\treedecomp$ with $B_x = \emptyset$.
\item \textbf{Introduce vertex bag}: an~internal vertex~$x$ of $\treedecomp$ 
with one child vertex~$y$ for which $B_x = B_y \cup \{v\}$ 
for some $v \notin B_y$. 
This bag is said to \emph{introduce} $v$.
\item \textbf{Introduce edge bag}: an~internal vertex~$x$ of $\treedecomp$ labeled 
with an~edge $uv \in E$ with one child bag~$y$ for which 
$u,v \in B_x = B_y$. 
This bag is said to \emph{introduce} $uv$.
\item \textbf{Forget bag}: an~internal vertex~$x$ of $\treedecomp$ with one child 
bag~$y$ for which $B_x = B_y \setminus \{v\}$ for some $v \in B_y$. 
This bag is said to \emph{forget} $v$.
\item \textbf{Join bag}: an~internal vertex $x$ with two child vertices 
$y$ and~$z$ with $B_x = B_y = B_z$.
\end{itemize}
We additionally require that every edge in $E$ is introduced exactly once. 
\end{definition}
The main differences between standard nice tree decompositions used by Kloks~\cite{Kloks94} and~this notion are: emptiness of leaf and~root bags and~usage of introduce edge bags.

As Cygan et al. observed in~\cite{my}, given an~arbitrary tree decomposition, a~nice tree decomposition of the same width can be found in polynomial time. Therefore, we can assume that all our algorithms are given a~tree decomposition that is nice.

Having fixed the root $r$, we associate with each node $x$ 
of a~tree decomposition $\treedecomp$ a~set $V_x \subseteq V$, where a~vertex 
$v$ belongs to $V_x$ iff there is a~bag $y$ which is a~descendant of $x$ in $\treedecomp$ with $v \in B_y$ (we follow convention that $x$ is its own descendant). We also associate with each bag $x$ of $\treedecomp$ a~subgraph of $\subgraph{x}$ defined as follows:
$$ \subgraph{x} = \left(\subbags{x}, \subedges{x} = \{e\ |\  \textrm{$e$ is introduced in a~descendant of $x$ } \}\right). $$
As every edge is introduced exactly once, for each join bag $x$ with children $y,z$, $\subedges{x}$ is a~disjoint sum of $\subedges{y}$ and~$\subedges{z}$.

\section{The model of logic}\label{sec:logic}

We begin with introducing a~notion of a~{\emph{finitely recognizable set}}.

\begin{definition}
A set $S\subseteq \N$ is called {\emph{finitely recognizable}} iff there exists a~finite monoid $M$, a~set $F\subseteq M$ and~homomorphism $\alpha_S: \N \to M$ such that $S=\alpha_S^{-1}(F)$.
\end{definition}

The notion of finitely recognizable sets coincides with semilinear sets over $\N$. To better understand the intuition behind it, let us state following simple fact. 

\begin{lemma}\label{lem:finrec}
A set $S\subseteq \N$ is finitely recognizable iff it is ultimately periodic, i.e. there exist positive integers $N,k$ such that $n\in S \Leftrightarrow n+k\in S$ for all $n\geq N$.
\end{lemma}

The fact can be considered a~folklore, however for the sake of completeness the proof can be found in Appendix~\ref{sec:logic-proof}.

Intuitively, the main property of finitely recognizable sets that will be useful, is that one can represent the behaviour of a~nonnegative integer with respect to the operation of addition by one of finitely many values --- the elements of the monoid. 

Now, we are ready to introduce the syntax and~semantics of \lagl. We will do this in two steps. First, we introduce the inner, modal part of the syntax. Then, we explain how this part is to be put into the context of quantification over subsets of vertices and~edges, thus creating a~framework for defining computational problems.

\subsection{The inner logic}

The inner logic will be called \cmglfull (\cmgl). A~formula $\psi$ of \cmgl is evaluated in a~certain vertex $v$ of a~(directed) graph $G$ supplied by a~vector of subsets of vertices $\vX$ and~a~vector of subsets of edges $\vY$, of length $p,q$ respectively. If $\psi$ is true in vertex $v$ of graph $G$, we will denote it by $G,\vX,\vY,v\models \psi$. We begin with the syntax of \cmgl for undirected graphs, defined by the following grammar:
\begin{align*}
\psi := & \  \neg \psi \  |\  \psi \wedge \psi \  |\   \psi \vee \psi \  |\   \psi \Rightarrow \psi \  |\   \psi \Leftrightarrow \psi \  | \  \mathbb{X} \  |\   \mathbb{Y} \  |\   \dia^S \psi \  |\  \rect^S \psi \\
\mathbb{X} := & \  X_1 \  |\  X_2 \  |\  \ldots \  |\  X_p\\
\mathbb{Y} := & \  Y_1 \  |\  Y_2 \  |\  \ldots \  |\  Y_q
\end{align*}
The boolean operators are defined naturally. Let us firstly discuss the modal quantifiers $\dia^S$ and~$\rect^S$. By definition, $S$ has to be a~finitely recognizable set. We define the semantics of $\dia^S$ in the following manner: we say that $G,\vX,\vY,v\models \dia^S \psi$ iff the number of neighbours $w$ of vertex $v$ satisfying $G,\vX,\vY,w\models \psi$ belongs to $S$. The quantifier $\rect^S$ is somewhat redundant, as we say that $G,\vX,\vY,v\models \rect^S \psi$ iff $G,\vX,\vY,v\models \neg \dia^S \neg \psi$. To shorten notation we use $\dia$ for $\dia^{\N^+}$ and~$\rect$ for $\rect^{\N^+}$, where $\N^+$ is the set of positive integers. Thus, the definitions of $\dia$ and~$\rect$ coincide with the natural way of introducing these quantifiers in other modal logics: $\dia \psi$ means that $\psi$ has to be true in at least one neighbour, while $\rect \psi$ means that $\psi$ has to be true in all the neighbours. Observe that the evaluation of the formula can be viewed as a~process of walking on the graph --- each time we evaluate a~modal quantifier we move to a~neighbour of the current vertex. Thus, after the first modal quantification there is a~well specified edge that was used to directly access the current vertex from his neighbour.

Operators $\mathbb{X}$ can be viewed as unary predicates, checking whether the vertex, in which the formula is evaluated, belongs to a~particular $X_i$. Formally, $G,\vX,\vY,v\models X_i$ iff $v\in X_i$. Operators $\mathbb{Y}$ play the same role for edges --- they check, whether the edge that was used to directly access the vertex belongs to a~particular $Y_j$. Therefore, we narrow ourselves only to such formulas that use operators $\mathbb{Y}$ under some quantification.

We extend the logic to directed graphs by defining the neighbour to be a~vertex that is adjacent via an~arc, with no matter which direction. We introduce two new operators belonging to $\mathbb{Y}$: $\downarrow$ and~$\uparrow$. The $\downarrow$ operator is true iff the arc that was used to directly access the current vertex is directed towards it, while $\uparrow$ is true iff it is directed towards the neighbour. Note that the new operators are significantly different from other operators in $\mathbb{Y}$, as they are not symmetrical from the point of view of the endpoints.

\begin{remark}
In order to define the semantics of \cmgl properly, without awkwardness of edge operators, we could bind them to the model quantifiers. In this variation of \cmgl, modal quantifiers are defined as $\dia^S_\beta\psi,\rect^S_\beta\psi$ for $\beta$ being a~boolean combination of operators from $\mathbb{Y}$. The lower indices of quantifiers are the only place operators from $\mathbb{Y}$ can occur. The semantics of diamond is now defined as following: $\dia^S_\beta\psi$ is true in $v$ iff the number of edges $vw$ satisfying $\beta$, such that $\psi$ is satisfied in $w$, belongs to $S$. $\rect^S_\beta\psi$ is defined to be equivalent to $\neg\dia^S_\beta\neg\psi$. It is not hard to transform a~\cmgl formula to an~equivalent formula of this form. Having expressed all boxes by diamonds, in bottom-up manner we transform every subformula $\dia^{S_i}\psi_i$ to a~form $\dia^{S_i}\bigwedge_{j=1}^{2^q}(\beta_j\Rightarrow \gamma_j)$, where $\beta_j$ are conjunctions of $\mathbb{Y}$ operators and~their negations, expressing all possible alignments of the edge to sets $\eY_j$, while $\gamma_j$ use only $\mathbb{X}$ operators and~subformulas beginning with quantification. Obtained formula is however equivalent to a~formula
$$\bigvee_{(m_j)_{j=1}^{2^q}\ :\ \sum_{j=1}^{2^q}m_j\in F_i}\quad \bigwedge_{j=1}^{2^q}\quad \dia^{S^{m_j}_i}_{\beta_j} \gamma_j$$
for $S^{m_j}_i=\alpha_i^{-1}(m_j)$, where $S_i=\alpha_i^{-1}(F_i)$ for $\alpha_i$ being a~homomorphism mapping $\N$ into a~finite monoid $M_i$. The described variation is a~cleaner form of \cmgl, however it is much less convenient for expressing actual computational problems. 
\end{remark}

\subsection{The outer logic}\label{sec:outer}

Let an~{\emph{instance}} be a~quadruple $(G,\vFX,\vFY,\overline{k})$: a~(directed) graph $G=(V,E)$ together with a~vector of fixed subsets of vertices $\vFX$, a~vector of fixed subsets of edges $\vFY$ and~a~vector of integer parameters $\overline{k}$. In most cases the fixed sets are not used, however they can be useful to distinguish subsets of vertices or~edges of the graph that are given in the input, like, for example, terminals in the \steinertree problem. Let $\mathcal{K}$ be a~class of instances: a~set of instances with the same lengths of vectors $\vFX,\vFY,\overline{k}$. We say that $\mathcal{K}$ is expressible in \lagl iff belonging to $\mathcal{K}$ is equivalent to satisfying a~fixed formula $\ff$ of the following form:
$$\ff=\exists_{\vX}\exists_{\vY} \left [\phi \wedge \forall_v G,\vFX,\vFY,\vX,\vY,v\models \psi\right].$$
Here:
\begin{itemize}
\item $\vX$ and~$\vY$ are vectors of quantified subsets of vertices and~edges respectively;
\item $\phi$ is an~arbitrary quantifier-free arithmetic formula over the parameters, cardinalities of sets of vertices and~edges of $G$ and~cardinalities of fixed and~quantified sets;
\item $\psi$ is a~\cmgl formula evaluated on the graph $G$ supplied with all the fixed and~quantified sets.
\end{itemize}
We say that such formulas belong to \laglfull (\lagl).

\begin{example}
The \vertexcover problem, given an~undirected graph $G$ and~an~integer $k$, asks whether there exists a~set of at most $k$ vertices such that every edge has at least one endpoint in the set. This can be reformulated as following: if a~vertex is not chosen, then all its neighbours have to be chosen. Thus, the class of YES instances of \vertexcover can be expressed in \lagl using the following formula:
$$\exists_{X\subseteq V} (|X|\leq k) \wedge \forall_v G,X,v\models (\neg X \Rightarrow \rect X).$$
\end{example}

\begin{example}
The \rdomset problem, given an~undirected graph $G$ and~an~integer $k$, asks whether there exists a~set of at most $k$ vertices such that every vertex is at distance at most $r$ from a~vertex from the set. The class of YES instances of \rdomset can be expressed in \lagl using the following formula:
$$\exists_{X\subseteq V} (|X|\leq k) \wedge \forall_v G,X,v\models \underbrace{(X \vee \dia (X \vee \dia (X \vee \ldots \dia (X \vee \dia X) \ldots )))}_{r \text{ quantifications}}.$$
\end{example}

\section{Tractability of problems expressible in \lagl}\label{sec:tractability}

We are ready to prove the main result of the paper, namely the tractability of \krecognition problem. The algorithm will base on the technique of prediction, useful in the construction of more involved dynamic programming routines on various types of decompositions. For an~example, see the tractability result of Demaine et al. for \rdomset \cite{fomin} that is in fact a~prototype of the constructed algorithm.

\defproblemu{\krecognition}{An instance $I=(G,\vFX, \vFY, \overline{k})$}{Does $I\in \mathcal{K}$?}
\vspace{-0.2cm}
\begin{theorem}\label{thm:deterministic}
If the class of instances $\mathcal{K}$ is expressible in \lagl, then there exists an~algorithm that, given an~instance $I$ along with a~tree decomposition of $G$ of width $t$, solves \krecognition in time $c^t |G|^{O(1)}$ for some constant $c$. Moreover, the algorithm can also compute the number of vectors $\vX,\vY$ satisfying the formula $\ff$ defining $\mathcal{K}$.
\end{theorem}

\begin{proof}
As was already mentioned in Section \ref{sec:preliminaries}, we may assume that the given tree decomposition is a~nice tree decomposition.

Let $\ff=\exists_{\vX}\exists_{\vY} \left [\phi \wedge \forall_v G,\vFX,\vFY,\vX,\vY,v\models \psi\right]$ be the formula defining the class $\mathcal{K}$ of instances of form $(G,\vFX,\vFY,\overline{k})$. Denote by $p_0,q_0,p_1,q_1$ lengths of vectors $\vFX$, $\vFY$, $\vX$, $\vY$ respectively. We show the algorithm for computing the number of possible solutions $\vX,\vY$; testing the outcome against zero solves the decision problem.

Firstly, the algorithm counts the cardinalities of fixed sets from vectors $\vFX,\vFY$. Then it introduces these constants into the arithmetic formula $\phi$ along with the parameters and~the numbers of vertices and~edges of $G$. Now, the algorithm branches into $(1+|V|)^{p_1}(1+|E|)^{q_1}$ subroutines: in each it fixes the expected cardinalities of quantified sets from vectors $\vX,\vY$. The algorithm executes only the branches with cardinalities satisfying $\phi$ and~at the end sums up obtained numbers of solutions. This operation yields only a~polynomial blow-up of the running time, so we may assume that the expected cardinalities of all the quantified sets are precisely determined. Let us denote by $\eiterV$, $\eiterE$ vectors of expected cardinalities of $\vX$, $\vY$ respectively.

As $\rect^S$ quantifier can be expressed by $\dia^S$ quantifier, we may assume that $\psi$ uses only $\dia^S$ quantifiers. Consider all subformulas $\psi_1,\psi_2,\ldots,\psi_l$ of $\psi$ beginning with a~quantifier, denote $\psi_i=\dia^{S_i} \psi'_i$. Let $S_i$ be defined as $\alpha_{S_i}^{-1}(F_i)$ for homomorphism $\alpha_{S_i}:\N \to M_i$, finite monoid $M_i$ and~$F_i\subseteq M_i$. Let $\Hh=\prod_{i=1}^l M_i$ be a~product monoid.

Let us denote $\Xx=\{0,1\}^{p_1}$, $\Pp=\{0,1\}^l$. Furthermore, let $\Ii=\Hh\times \Pp \times\Xx$. Intuitively, $\Ii$ is a~set of possible information that can be stored about a~vertex. The information consists of: {\emph{history}}, an~element of $\Hh$; {\emph{prediction}}, a~binary vector from $\Pp$ indicating, which formulas $\psi_i$ are predicted to be true in a~vertex; and~{\emph{alignment}}, a~binary vector from $\Xx$ indicating, to which quantified sets $\eX_i$ a~vertex belongs.

Before we proceed to the formal description of the algorithm, let us give some intuition about what will be happening. The history is an~element of the product monoid, used to count already introduced neighbours satisfying certain formulas $\psi'_i$. The additive structure on $\Hh$ enables us to update the history during introduce edge and~join steps. However, while determining satisfaction of subformulas $\psi_i$ in vertices of the graph, for the vertices in the bag we have to know their 'type' in the whole graph, not just the influence of already introduced part. Therefore, we introduce prediction: the information, which subformulas are predicted to be true in a~vertex in the whole graph. When doing updates while introducing edges we can access the predicted values, however when forgetting a~vertex we have to ensure that its history is consistent with the prediction.

Let $\cand$ be the set of solutions, i.e., pairs of vectors $\vX,\vY$ for which $\psi$ is satisfied in every vertex and~satisfying constraints imposed on cardinalities of the sets. For a~node $x$ of the tree decomposition let $s\in \Ii^{\bag{x}}$ be an~{\emph{information evaluation}}. We denote $s(v)=(h_v,\pi_v,b_v)$, where $v\in \bag{x}$. Let $\iterV$, $\iterE$ be vectors of integers of lengths $p_1,q_1$ respectively, satisfying $0\leq \iterv_i\leq \eiterv_i$ and~$0\le \itere_j\leq \eitere_j$ for all $1\leq i\leq p_1$, $1\leq j\leq q_1$. Let us define $\cand_x(\iterV,\iterE,s)$: the set of partial solutions consistent with vectors $\iterV,\iterE$ and~information evaluation $s$. By this we mean the set of pairs of vectors $\vX,\vY$ of subsets of vertices and~edges of $\subgraph{x}$ respectively, such that the following conditions are satisfied.
\begin{itemize}
\item $|X_i|=\iterv_i$ for $1\leq i\leq p_1$, $|Y_j|=\itere_j$ for $1\leq j\leq q_1$. 
\item Every $v\in \bag{x}$ belongs to exactly those $\eX_i$, for which the $i$-th coordinate of $b_v$ is $1$.
\item In all $v\in \subgraph{x}\setminus\bag{x}$ the formula $\psi$ is satisfied, when evaluated in $\subgraph{x}$ supplied with quantified and~fixed sets. However, when evaluating some formula $\psi_j$ in a~vertex $w\in \bag{x}$ we access the corresponding coordinate in the prediction $\pi_w$ instead of actually evaluating it in $\subgraph{x}$.
\item For all $v\in \bag{x}$ the number of neighbours of $v$ in $\subgraph{x}$ satisfying the formula $\psi'_i$ maps in $\alpha_{S_i}$ to the $i$-th coordinate of $h_v$, where $\psi'_i$ is evaluated in the neighbour as if it was accessed directly from $v$. Again, boolean values of formulas $\psi_j$ in the vertices of $\bag{x}$ are taken from the prediction instead of truly evaluated.
\end{itemize}
Observe that $\cand=\cand_r(\overline{\mathbf{x}},\overline{\mathbf{y}},\emptyset)$. The number of possible vectors $\iterV,\iterE$ and~information evaluations $s$ is bounded by $|\Ii|^t|G|^{O(1)}$, so it suffices to show a~dynamic program that computes $A_x(\iterV,\iterE,s)=|\cand_x(\iterV,\iterE,s)|$ for all possible arguments in a~bottom-up fashion. It is not hard to implement the performance of the routine for every type of a~bag. The details of an~algorithm running in $|\Ii|^{2t}|G|^{O(1)}$ time are described in Appendix~\ref{sec:tractability-details}.

\end{proof}

\section{Adding connectivity requirements}\label{sec:connectivity}

We extend \lagl by connectivity requirements. We say that an~arithmetic formula $\phi(\overline{x},y)$ is {\emph{monotone}} over $y$ iff $\phi(\overline{x},y)\Rightarrow \phi(\overline{x},y')$ for $y\geq y'$. In \laglc, the arithmetic formula $\phi$ can also depend on $|\compnomark(\eFX_i)|$, $|\compnomark(\eFY_j)|$, $|\compnomark(\eX_i)|$, $|\compnomark(\eY_j)|$, vectors of numbers of connected components of fixed and~quantified sets. The dependence on the quantified part, variables $|\compnomark(\eX_i)|$ and~$|\compnomark(\eY_j)|$, is however restricted to be monotone, i.e., if $y$ is the variable of $\phi$ that corresponds to the number of connected components of some quantified set, then $\phi$ has to be monotone over $y$. The need of monotonicity can be justified by a~number of lower bounds for problems involving maximization of the number of connected components, due to Cygan et al.~\cite{my}.

It appears that we can combine the Cut\&Count technique with the dynamic programming routine described in Section~\ref{sec:tractability} in order to obtain similar tractability of problems defined in \laglc. Unfortunately, application of the technique gives us the tractability of only the decision problem. To the best of author's knowledge, extending the Cut\&Count technique to counting problems is an~open question, posted in~\cite{my}.

\begin{theorem}\label{thm:randomized}
If the class of instances $\mathcal{K}$ is expressible in \laglc, then there exists a~Monte-Carlo algorithm that, given the instance $I$ along with a~tree decomposition of $G$ of width $t$, solves \krecognition in time $c^t |G|^{O(1)}$ for some constant $c$. The algorithm cannot produce false positives and~produces false negatives with probability at most $\frac{1}{2}$. 
\end{theorem}

The proof is a~quite straightforward translation of the proof of Theorem \ref{thm:deterministic} to the language of Cut\&Count. For the sake of completeness, it can be found in Appendix~\ref{sec:con-proof}.

\section{The necessity of acyclicity}\label{sec:negative}

We prove the intractability results for two expository non-acyclic problems.

\noindent \defproblemu{\lvertexdeletion}{An undirected graph $G$ and~an~integer $k$}{
	Is it possible to remove at most $k$ vertices from $G$ so that the remaining
	vertices induce a~graph without cycles of length $l$?}
\noindent \defproblemu{\ltlvertexdeletion}{An undirected graph $G$ and~an~integer $k$}{
	Is it possible to remove at most $k$ vertices from $G$ so that the remaining
	vertices induce a~graph without cycles of length at most $l$?}
\vspace{-0.5cm}
\begin{theorem}\label{thm:negative}
Assuming ETH, there is no $2^{o(p^2)} |G|^{O(1)}$ time algorithm for \lvertexdeletion nor for \ltlvertexdeletion for any $l\geq 5$.
The parameter $p$ denotes the width of a~given path decomposition of the input graph.
\end{theorem}
As a~path decomposition of width $p$ is also a~tree decomposition of width $p$, the result is in fact stronger than analogous for treewidth instead of pathwidth. Before we proceed to the proof, note that both these problems admit a~simple $2^{O(t^2)} |G|^{O(1)}$ dynamic programming algorithm, where $t$ is the width of a~given tree decomposition. In the state, one remembers for every pair of vertices of bag $\bag{x}$, whether in $\subgraph{x}$ they can be connected via paths of length $1,2,\ldots,l-1$ disjoint with the solution.

We present a~polynomial-time reduction that given a~\tcnfsat instance: a~formula $\ff$ in \tcnf over $n$ variables and~consisting of $m$ clauses, produces a~graph $G$ along with its path decomposition of width $O(\sqrt{n})$ and~an~integer $k$, such that
\begin{itemize}
\item if $\ff$ is satisfiable then $(G,k)$ is a~YES instance of \ltlvertexdeletion;
\item if $(G,k)$ is a~YES instance of \lvertexdeletion then $\ff$ is satisfiable.
\end{itemize}
As every YES instance of \ltlvertexdeletion is also a~YES instance of \lvertexdeletion, the constructed instance $(G,k)$ is equivalent to given instance of \tcnfsat both when considered as an~instance of \lvertexdeletion and~of \ltlvertexdeletion. Thus, existence of an~algorithm for \lvertexdeletion or~\ltlvertexdeletion running in $2^{o(p^2)} |G|^{O(1)}$ time would yield an~algorithm for \tcnfsat running in $2^{o(n)} (n+m)^{O(1)}$ time, contradicting ETH. We can assume that each clause in $\ff$ contains exactly three literals by copying some of them if necessary.

Let us choose $\alpha=\lfloor \frac{l-1}{2}\rfloor$, $\beta=\lceil \frac{l+1}{2} \rceil$. Thus, following conditions are satisfied:
$2\leq \alpha< \beta$, $\alpha + \beta = l$, $2\beta>l$, $2\alpha+4>l$.

Now we show the construction of the instance. The proof of its soundness and~the bound on pathwidth can be found in Appendix~\ref{sec:negative-details}.
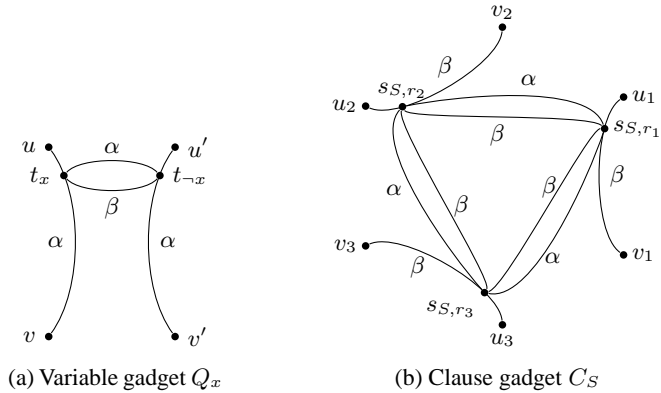
\vspace*{-0.5cm}
\begin{figure}
\begin{center}
\subfloat[Variable gadget $Q_x$]{
\begin{tikzpicture}[scale=0.42]
   \tikzstyle{vertex}=[circle,fill=black, minimum size=0.1cm,inner sep=0pt]

   \node[vertex] (u1) at (-2,3) {};
   \node[vertex] (v1) at (-2,-3) {};
   \node[vertex] (u2) at (2,3) {};
   \node[vertex] (v2) at (2,-3) {};

   \node[vertex] (tx) at (-1.52,2.1) {};
   \node[vertex] (tnx) at (1.52,2.1) {};

  \draw (u1) .. controls (-0.9,1.5) and (-0.9,-1.5) .. (v1);
  \draw (u2) .. controls (0.9,1.5) and (0.9,-1.5) .. (v2);
  \draw (tx) .. controls (-1,2.7) and (1,2.7) .. (tnx);
  \draw (tx) .. controls (-1,1.5) and (1,1.5) .. (tnx);

      \draw[left] (u1.west) node {$u$};
      \draw[left] (v1.west) node {$v$};
      \draw[left] (tx.west) node {$t_x$};
      \draw[right] (u2.east) node {$u'$};
      \draw[right] (v2.east) node {$v'$};
      \draw[right] (tnx.east) node {$t_{\neg x}$};

      \draw[left] (-1.26,0.0) node {$\alpha$};
      \draw[right] (1.26,0.0) node {$\alpha$};
      \draw[above] (0.0,2.54) node {$\alpha$};
      \draw[below] (0.0,1.66) node {$\beta$};

\end{tikzpicture}
}
\qquad \qquad 
\subfloat[Clause gadget $C_S$]{
\begin{tikzpicture}[scale=0.42]
   \tikzstyle{vertex}=[circle,fill=black, minimum size=0.1cm,inner sep=0pt]

    \begin{scope}[rotate=0]
        \node[vertex] (x1) at (4,2.5) {};
        \node[vertex] (y1) at (4,-2.5) {};
	\draw (x1) .. controls (3,1.8) and (3,-1.8) .. (y1);
	\draw[right] (3.3,0) node {$\beta$};
	\node[vertex] (z1) at (3.4,1.5) {};
	\draw[right] (x1) node {$u_1$};
	\draw[right] (y1) node {$v_1$};
	\draw[right] (z1) node {$s_{S,r_1}$};	
	\node (pz11) at (2.9,3.0) {};
	\node (pz12) at (2.9,2.0) {};
	\node (pz13) at (2.9,0.0) {};
	\node (pz14) at (2.9,1.5) {};
	\draw[above] (1.0,2.6) node {$\alpha$};
	\draw[below] (0.0,1.9) node {$\beta$};
    \end{scope}
    \begin{scope}[rotate=120]
        \node[vertex] (x2) at (4,2.5) {};
        \node[vertex] (y2) at (4,-2.5) {};
	\draw (x2) .. controls (3,1.8) and (3,-1.8) .. (y2);
	\draw[above] (3.3,0) node {$\beta$};
	\node[vertex] (z2) at (3.4,1.5) {};
	\draw[left] (x2) node {$u_2$};
	\draw[above] (y2) node {$v_2$};
	\draw[above] (z2) node {$s_{S,r_2}$};	
	\node (pz21) at (2.9,3.0) {};
	\node (pz22) at (2.9,2.0) {};
	\node (pz23) at (2.9,0.0) {};
	\node (pz24) at (2.9,1.5) {};
	\draw[left] (1.0,2.6) node {$\alpha$};
	\draw[right] (0.0,1.9) node {$\beta$};
    \end{scope}
    \begin{scope}[rotate=240]
        \node[vertex] (x3) at (4,2.5) {};
        \node[vertex] (y3) at (4,-2.5) {};
	\draw (x3) .. controls (3,1.8) and (3,-1.8) .. (y3);
	\draw[left] (3.5,-0.3) node {$\beta$};
	\node[vertex] (z3) at (3.4,1.5) {};
	\draw[below] (x3) node {$u_3$};
	\draw[left] (y3) node {$v_3$};
	\draw[below left] (z3) node {$s_{S,r_3}$};	
	\node (pz31) at (2.9,3.0) {};
	\node (pz32) at (2.9,2.0) {};
	\node (pz33) at (2.9,0.0) {};
	\node (pz34) at (2.9,1.5) {};
	\draw[below] (1.0,2.6) node {$\alpha$};
	\draw[above] (0.0,1.9) node {$\beta$};
    \end{scope}
	\draw (z1) .. controls (pz11) and (pz23) .. (z2);
	\draw (z1) .. controls (pz12) and (pz24) .. (z2);
	\draw (z2) .. controls (pz21) and (pz33) .. (z3);
	\draw (z2) .. controls (pz22) and (pz34) .. (z3);
	\draw (z3) .. controls (pz31) and (pz13) .. (z1);
	\draw (z3) .. controls (pz32) and (pz14) .. (z1);
\end{tikzpicture}
}
\end{center}
\end{figure}

\vspace*{-1.5cm}
\subsubsection*{Construction.} We begin the construction by creating two sets of vertices $A,B$, each consisting of $\left\lceil \sqrt{2n}\right\rceil$ vertices. As~$|A\times B|\geq 2n$, let us take any injective function $\psi: L \to A\times B$, where $L$ is the set of literals over the variables of the formula $\ff$, i.e., symbols $x$ and~$\neg x$ for all variables $x$.

For every variable $x$ we construct a~{\emph{variable gadget}} $Q_x$ in the following manner. Let $\psi(x)=(u,v)$ and~$\psi(\neg x)=(u',v')$ ($u$~and~$u'$ or~$v$~and~$v'$ may possibly coincide). Connect $u$ with $v$ and~$u'$ with $v'$ via paths of length $\alpha$. Denote the inner vertices of the paths that are closest to $u$~and~$u'$ by $t_{x}$~and~$t_{\neg x}$ respectively. Connect $t_{x}$ with $t_{\neg x}$ via two paths: one of length $\alpha$ and~one of length $\beta$. Note that these two paths form a~cycle of length $l$. The gadget consists of all the constructed paths along with vertices $u,u',v,v'$.

Now, for every clause $S=r_1\vee r_2\vee r_3$, where $r_1,r_2,r_3$ are literals, we construct the {\emph{clause gadget}} $C_S$ in the following manner. Let $\psi(r_i)=(u_i,v_i)$ for $i=1,2,3$ ($u_i$~or~$v_i$ may possibly coincide). For $i=1,2,3$ connect $u_i$ with $v_i$ via a~path of length $\beta$, and~denote inner vertices of these paths that are closest to $u_i$ by $s_{S,r_i}$. Connect each pair $(s_{S,r_1},s_{S,r_2})$, $(s_{S,r_2},s_{S,r_3})$, $(s_{S,r_3},s_{S,r_1})$ via two paths: one of length $\alpha$ and~one of length $\beta$. Thus, we connect $s_{S,r_1},s_{S,r_2},s_{S,r_3}$ by a~triple of cycles of length $l$. The gadget consists of all the constructed paths together with vertices $u_i,v_i$.

We conclude the construction by setting $k=n+2m$.

\section{Conclusions and~open problems}\label{sec:conclusions}

In this paper we introduced a~logical formalism based on modality, \laglfull, capturing majority of problems known to be tractable in single exponential time when parameterized by treewidth. We proved that testing, whether a~fixed \lagl formula is true in a~given graph, admits an~algorithm with complexity $c^t|G|^{O(1)}$, where $t$ is the width of given tree decomposition. We extended \lagl by connectivity requirements and~obtained a~similar tractability result using the Cut\&Count technique of Cygan et al.~\cite{my}. The need of modality of the logic was justified by a~negative result under ETH that two model problems with non-acyclic requirements are not solvable in $2^{o(p^2)}|G|^{O(1)}$, where $p$ is the width of a~given path decomposition.

One open question is to breach the gap in the presented negative result. For $l=3$, \lvertexdeletion is solvable in single exponential time in terms of treewidth, while for $l\geq 5$ our negative result states that such a~robust solution is unlikely. To the best of author's knowledge, for $l=4$ there are no matching lower and~upper bounds.

Secondly, there are problems that admit a~single exponential algorithm when parameterized by treewidth, but are not expressible in \lagl. One example could be \lcliquevertexdeletion, that, given a~graph $G$ along with an~integer $k$, asks whether there exists a~set of at most $k$ vertices that hits all the subgraphs $K_l$. A~dynamic program for this problem running in time $4^t|G|^{O(1)}$ can be constructed basing on the observation, that for every subclique of a~graph there has to be a~bag fully containing it. Can we find an~elegant extension of \lagl that would capture also such type of problems?

\subsubsection*{Acknowledgments}
The author would like to thank Miko\l{}aj Boja\'nczyk for invaluable help with the logical side of the paper, as well as Marek Cygan, Marcin Pilipczuk and~Jakub Onufry Wojtaszczyk for many helpful comments on the algorithmic part.

\bibliographystyle{plain}
\bibliography{cut-logic}

\newpage

\appendix

\section{Proof of Lemma~\ref{lem:finrec}}\label{sec:logic-proof}

\begin{lemma}[Lemma \ref{lem:finrec}, restated]
A set $S\subseteq \N$ is finitely recognizable iff it is ultimately periodic, i.e. there exist positive integers $N,k$ such that for all $n\geq N$ the following holds: $n\in S \Leftrightarrow n+k\in S$.
\end{lemma}
\begin{proof}
Assume that $S$ is finitely recognizable. Let $M$ be a~finite monoid and~$\alpha_S: \N \to M$ a~homomorphism such that $S=\alpha_S^{-1}(F)$ for some $F\subseteq M$. Recall that for every finite monoid $N$ there exists such a~number $\omega$, called the {\em{idempotent power}}, that for every $a\in N$ the element $\underbrace{a+\ldots+a}_{\omega}$ is an~idempotent, i.e. $\underbrace{a+\ldots+a}_{2\omega}=\underbrace{a+\ldots+a}_{\omega}$. Let $\omega$ be the idempotent power in $M$. We claim that we can take $N=k=\omega$. Indeed, if $n\geq \omega$, then 
\begin{align*}
\alpha_S(n)=\underbrace{\alpha_S(1)+\ldots+\alpha_S(1)}_{n}=\underbrace{\alpha_S(1)+\ldots+\alpha_S(1)}_{\omega}+\underbrace{\alpha_S(1)+\ldots+\alpha_S(1)}_{n-\omega}=\\ 
\underbrace{\alpha_S(1)+\ldots+\alpha_S(1)}_{2\omega}+\underbrace{\alpha_S(1)+\ldots+\alpha_S(1)}_{n-\omega}=\alpha_S(n+\omega).
\end{align*}

Now assume that we have positive integers $N,k$ such that $n\in S\Leftrightarrow n+k\in S$ for $n\geq N$. Take $M$ to be a~monoid over $\{0,1,\ldots,N+k-1\}$ with the operation $+_M$ defined as follows:
$$a+_Mb =\begin{cases} a+b & \quad \hbox{ if } a+b<N+k,\\
N+((a+b-N) \hbox{ mod } k) & \quad \hbox{otherwise.}
\end{cases}$$
It is easy to verify that it is indeed a~monoid. Furthermore, let us define homomorphism $\alpha_S: \N\to M$ by setting $\alpha_S(0)=0, \alpha_S(1)=1$ and~extending it naturally. A~straightforward check proves that $S=\alpha_S^{-1}(\{0,1,\ldots,N+k-1\}\cap S)$.
\end{proof}

\section{Details of the dynamic program from the proof of Lemma~\ref{lem:finrec}}\label{sec:tractability-details}

We present, how the computation of values $A_x(\iterV,\iterE,s)$ should be performed for every type of a~bag in a~bottom-up fashion, in order to obtain a~$|\Ii|^{2t}|G|^{O(1)}$ algorithm.

The lenghts binary representations of values $A_x(\iterV,\iterE,s)$ are bounded by a~polynomial in the size of input, hence arithmetic operations during the computation can be carried out in polynomial time. We follow convention that all values $A_x$ with improper arguments, for example having negative coordinates, are defined to be zeroes. For a~condition $c$, by $[c]$ we denote $1$ if $c$ is true, and~$0$ otherwise. Moreover, for a~function $s$ by $s[v \to \alpha]$ we denote the function $s \setminus \{(v,s(v))\} \cup \{(v,\alpha)\}$. Note that this definition is correct even when $s$ is not defined on $v$.

\vskip 0.2cm
\noindent \textbf{Leaf bag $x$}:
	$$ A_x(\overline{0},\overline{0},\emptyset) = 1 $$
	$\overline{0}$ is a~vector of zeroes of appropriate length. All other values $A_x(\iterV,\iterE,\emptyset)$ are zeroes.
\vskip 0.2cm
\noindent\textbf{Introduce vertex $v$ bag $x$, with child $y$}:
	$$ A_x(\iterV,\iterE,s[v \to (h,\pi,b)]) = [h=e_\Hh]A_y(\iterV-b,\iterE,s) $$
	Observe that the introduced vertex has no neighbours so far, so its history must be void, which is indicated by checking whether the assigned history is $e_\Hh$, the identity of the product monoid. However, its prediction and~alignment can be arbitrary.
\vskip 0.2cm
\noindent\textbf{Introduce edge (arc) $uv$ bag $x$, with child $y$}:
	$$ A_x(\iterV,\iterE,s) = \sum_{d\in\{0,1\}^{q_1}}\ \ \sum_{s'\in S'} A_y(\iterV,\iterE-d,s') $$
	The first summation corresponds to all possible ways of choosing the family of sets $\eY_i$ the introduced edge belongs to. In the second summation we sum over all information evaluations $s'$ such that $s'$ differs from $s$ only on histories of vertices $u,v$ in the following manner. The history $h_u$ of $u$ in $s$ is the history $h'_u$ of $u$ in $s'$, but with images of $1$ added on precisely these coordinates $j$, which correspond to formulas $\psi'_j$ satisfied in $v$, when accessed from $u$. The symmetrical condition holds for the history $h_v$ of $v$ in $s$ and~the history $h'_v$ of $v$ in $s'$. This condition corresponds to updating the history after introducing the edge. Observe that given the predictions on $u,v$ along with the information about the sets (fixed or~quantified) $u,v$ belong to and~the information about the sets (fixed or~quantified) the edge (arc) $uv$ belongs to, we can compute which formulas $\psi'_j$ are satisfied in $u$ when accessed from $v$ and~vice versa. The number of considered evaluations $s'$ is constant and~they can be enumerated in constant time, so the computation of a~single value takes constant time.
\vskip 0.2cm
\noindent\textbf{Forget vertex $v$ bag $x$, with child $y$}:
	$$ A_x(\iterV,\iterE,s) = \sum_{(h,\pi,b)\in \Gg} A_y(\iterV,\iterE,s[v\to (h,\pi,b)]) $$
	The summation corresponds to possible information stored in the vertex we are forgetting. In order to forget a~vertex without violating the definition of $\cand(\iterV,\iterE,s)$ we have to ensure that the prediction is consistent with the history and~that $\psi$ is satisfied in $v$. Therefore, the summation is carried out over a~set $\Gg$ of {\emph{good}} triples $(h,\pi,b)$, such that 
\begin{itemize}
\item $h_i\in F_{S_i}$ iff $\pi_i=1$, for all $1\leq i\leq l$; 
\item the alignment $b$ and~the satisfaction of formulas $\psi_i$ on quantification depth $0$ (stored in prediction) make the formula $\psi$ true. Note that there are no unquantified edge operators, so $b$ along with knowledge of fixed sets $\vFX$ suffices to compute this.
\end{itemize}
The set $\Gg$ can be determined in constant time, so the computation of a~single value takes constant time.
\vskip 0.2cm
\noindent\textbf{Join bag $x$, with children $y,z$}:
	$$ A_x(\iterV,\iterE,s) = \sum_{\iterV'+\iterV''=\iterV+\xi(\bag{x})} \ \ \sum_{\iterE'+\iterE''=\iterE}\ \ \sum_{s'+s''=s} A_y(\iterV',\iterE',s')A_z(\iterV'',\iterE'',s'') $$
	$s'+s''=s$ denotes that for all the vertices $v\in \bag{x}$ the predictions in $s(v),s'(v),s''(v)$ are equal, the alignments in $s(v),s'(v),s''(v)$ are equal and~the histories in $s'(v)$ and~$s''(v)$ sum up to the history in $s(v)$. The first summation corresponds to splitting the expected cardinalities of sets $\vX$ in solution, however we have to take care of double counting the elements of the bag $\bag{x}$ by adding to the right side $\xi(\bag{x})$, the sum of the alignments in $s$ over the bag $\bag{x}$. The second summation corresponds to splitting the expected cardinalities of sets $\vY$ in the solution. As~every edge is introduced exactly once, the sets $\subedges{y}$ and~$\subedges{z}$ are disjoint and~sum to $\subedges{x}$, so there is no problem with double counting the edges. Note that the number of summands considered so far is polynomial. The last summation corresponds to splitting the information. In order to be able to merge two partial solutions built under bags $y$, $z$, the alignments has to be the same in both bags $\bag{y}$, $\bag{z}$ as well as the predictions. However, the histories are defined basing on numbers of neighbours satisfying appropriate conditions and~therefore should be added. As~$(\subedges{y},\subedges{z})$ is a~partition of $\subedges{x}$, we avoid problems with double counting the neighbours. Observe that we can compute all the needed values $A_x(\iterV,\iterE,s)$ in $|\Ii|^{2t}|G|^{O(1)}$ at once. For every pair $(s',s'')$ there exists at most one information evaluation $s$ such that $s'+s''=s$. Having computed it in polynomial time for every pair $(s',s'')$, for every $s$ we can iterate through all the contributing pairs in order to evaluate the presented formula on $A_x(\iterV,\iterE,s)$. Thus, having fixed splitting of the acumulators, each pair is considered at most once.

The computation of a~single value takes constant time in leaf, introduce, introduce edge and~forget steps, while the join step can be carried out in $|\Ii|^{2t}|G|^{O(1)}$. As~in every step the algorithm computes $|\Ii|^t|G|^{O(1)}$ values and~there are polynomially many steps, the whole algorithm runs in $|\Ii|^{2t}|G|^{O(1)}$ time.

\section{Proof of Theorem \ref{thm:randomized}}\label{sec:con-proof}

Let us recall that the crucial probabilistic tool used in the Cut\&Count technique is the Isolation Lemma.

\begin{definition}
A function $\omega:U \rightarrow \mathbb{Z}$ \emph{isolates} a~set family 
$\mathcal{F} \subseteq 2^U$ if there is a~unique $S' \in \mathcal{F}$ with 
$\omega(S')= \min_{S \in \mathcal{F}}\omega(S)$. 
\end{definition}

For $X \subseteq U$, $\omega(X)$ denotes 
$\sum_{u \in X}\omega(u)$.

\begin{lemma}[Isolation Lemma, \cite{isolation}]
\label{lem:iso}
Let $\mathcal{F} \subseteq 2^U$ be a~set family over a~universe $U$ with 
$|\mathcal{F}|>0$.
For each $u \in U$, choose a~weight $\omega(u) \in \{1,2,\ldots,N\}$ 
uniformly and~independently at random.
Then
$$\mathtt{prob}[\omega \textnormal{ isolates } \mathcal{F}] \geq 1 - \frac{|U|}{N}$$
\end{lemma}

We will now merge the deterministic result from Theorem \ref{thm:deterministic} with the Cut\&Count technique in order to prove Theorem \ref{thm:randomized}. We follow notation from~\cite{my} in order to make the proof easier to understand for the reader already familiar with the basics of Cut\&Count. The algorithm will be an~extension of the algorithm given by Theorem \ref{thm:deterministic}, therefore we will constantly refer to the details of its proof.

\begin{theorem}[Theorem \ref{thm:randomized}, restated]
If the class of instances $\mathcal{K}$ is expressible in \laglc, then there exists a~Monte-Carlo algorithm that, given the instance $I$ along with a~tree decomposition of $G$ of width $t$, solves \krecognition in time $c^t |G|^{O(1)}$ for some constant $c$. The algorithm cannot produce false positives and~produces false negatives with probability at most $\frac{1}{2}$. 
\end{theorem}
\begin{proof}
As was already mentioned in Section \ref{sec:preliminaries}, we may assume that the given tree decomposition is a~nice tree decomposition.

We will follow the same notation as in the proof of Theorem \ref{thm:deterministic}: the class $\mathcal{K}$ is defined by the formula $\ff=\exists_{\vX}\exists_{\vY} \left [\phi \wedge \forall_v G,\vFX,\vFY,\vX,\vY,v\models \psi\right]$ and~we are given an~instance $(G,\vFX,\vFY,\overline{k})$ together with the tree decomposition of $G$ of width $t$. $p_0,q_0,p_1,q_1$ are lengths of vectors $\vFX$, $\vFY$, $\vX$, $\vY$ respectively.

Firstly the algorithm computes cardinalities and~numbers of connected components of fixed sets $\vFX,\vFY$. These constants along with the parameters and~the numbers of vertices and~edges of the graph are introduced into the arithmetic formula $\phi$. Then, the algorithm branches into $(1+|V|)^{p_1}(1+|E|)^{q_1} \cdot (1+|V|)^{p_1}(1+|V|)^{q_1}$ subroutines, in each fixing the expected cardinalities and~numbers of connected components of all the quantified sets. The algorithm executes only these branches, for which the formula $\phi$ is satisfied.  Thus, the number of branches is polynomial. Every branch will return a~false negative with probability at most $\frac{1}{2}$. Therefore, we independently run every branch a~logarithmic number of times in order to reduce the probability of a~false negative to at most $\frac{1}{2K}$, where $K$ is the number of branches executed. Using the union bound we can bound the probability of a~false negative of the whole algorithm by $\frac{1}{2}$.

Let us fix a~branch. Let $\eiterV$, $\eiterE$ be the vectors of expected cardinalities of sets $\vX$, $\vY$ respectively, and~$\eitercV$, $\eitercE$ be the vectors of expected numbers of connected components of $\vX,\vY$ respectively. Observe that the algorithm instead of deciding whether there exist sets $\eX_i$, $\eY_j$ satisfying $|\compnomark(\eX_i)|=\eitercv_i$, $|\compnomark(\eY_j)|=\eiterce_j$ for all $1\leq i\leq p_1$, $1\leq j\leq p_1$, can decide whether there exist sets $\eX_i$, $\eY_j$ satisfying $|\compnomark(\eX_i)|\leq \eitercv_i$ and~$|\compnomark(\eY_j)|\leq \eiterce_j$. Indeed, by the monotonicity of the formula $\phi$ if the algorithm finds out that the answer to this (easier to satisfy) question is positive, then the answer to the whole task is positive as well. Therefore, we can relax the constraint imposed on the numbers of connected components of quantified sets to inequalities. From now on we focus only on the case, when all the expected cardinalities of quantified sets are fixed and~the expected number of connected components of every quantified set is bounded by some fixed number.

We proceed to the description of the dynamic programming routine. By {\emph{cut}} of a~graph $G=(V,E)$ we mean a~pair $(V_1,V_2)$ such that $V_1\cup V_2=V$ and~$V_1\cap V_2=\emptyset$. Let us recall the notion of {\emph{consistently cut subgraph}}, the main ingredient of the Cut\&Count technique.

\begin{definition}[Definition 3.2 of \cite{my}]
A cut $(V_1,V_2)$ of an~undirected graph $G=(V,E)$ is \emph{consistent} if 
$u \in V_1$ and~$v \in V_2$ implies $uv \notin E$. 
A {\em consistently cut subgraph} of $G$ is a~pair $(X,(X_1,X_2))$ 
such that $(X_1,X_2)$ is a~consistent cut of $G[X]$.

Similarly, for a~directed graph $D=(V,A)$ a~cut $(V_1,V_2)$ is consistent
if $(V_1,V_2)$ is a~consistent cut in the underlying undirected graph.
A consistently cut subgraph of $D$ is a~pair $(X,(X_1,X_2))$ 
such that $(X_1,X_2)$ is a~consistent cut of the underlying undirected graph of $D[X]$.
\end{definition}
Observe that in this definition $X$ can be a~subset of vertices as well as a~subset of edges. In both cases $(X_1,X_2)$ is a~cut of $V(G(X))$, which can be a~proper subset of $V$.
\vskip 0.3cm
\noindent{\bf{The Cut part.}} We will use the concept of markers, used in more involved applications of the Cut\&Count technique. Let us define the family of {\emph{candidate solution}} $\cand$ as the set of quadruples $(\vX,\vY,\Xmarkers,\Ymarkers)$, where lengths of $\Xmarkers,\Ymarkers$ are $p_1,q_1$ respectively, such that
\begin{itemize}
\item $\psi$ is satisfied in every vertex of $G$ supplied with sets $\vFX, \vFY, \vX,\vY$;
\item sets $\vX,\vY$ satisfy the imposed conditions on their cardinalities (but not necessarily on the numbers of connected components);
\item $\xmarkers_i\subseteq \eX_i, |\xmarkers_i|\leq\eitercv_i$ for every $1\leq i\leq p_1$; 
\item $\ymarkers_j\subseteq \eY_j, |\ymarkers_j|\leq\eiterce_j$ for every $1\leq j\leq q_1$.
\end{itemize}
Sets $\xmarkers_i,\ymarkers_j$ are called {\em{marker sets}}.

Suppose that we are given a~weight function $\omega: (V\times \{0,1,2\}^{p_1}) \cup (E\times \{0,1,2\}^{q_1}) \to \{1,2,\ldots,N\}$ for $N=2|(V\times \{0,1,2\}^{p_1}) \cup (E\times \{0,1,2\}^{q_1})|$. This weight function will be fixed throughout the whole proof. Given a~candidate solution $(\vX,\vY,\Xmarkers,\Ymarkers)$ we can define weight of a~vertex as $\omega(v,\overline{\chi})$, where: $\chi_i=0$ if $v\notin \eX_i$; $\chi_i=1$ if $v\in \eX_i$ but $v\notin \xmarkers_i$; and~$\chi_i=2$ if $v\in \xmarkers_i$. Similarly we define the weight of an~edge (arc). Let us define the weight of a~candidate solution as the sum of weights over all the vertices and~edges (arcs) of the graph. Let $\cand_\targetW$ be the set of candidate solutions with weight exactly $\targetW$. Observe that the maximal value of $\targetW$ is $O(N(|V|+|E|))$, which is polynomial in terms of the input size.

Now we define the family of {\emph{solutions}} $\sols\subseteq \cand$ by requiring from a~candidate solution that
\begin{itemize}
\item each connected component of $G[\eX_i]$ contains at least one vertex from $\xmarkers_i$;
\item each connected component of $G[\eY_j]$ contains at least one edge (arc) from $\ymarkers_j$.
\end{itemize}
As $|\xmarkers_i|\leq\eitercv_i$ and~$|\ymarkers_j|\leq\eiterce_j$, these conditions imply the constraints on the numbers of connected components of sets $\vX,\vY$. Of course, as every pair $\vX,\vY$ satisfying all the imposed conditions can be marked appropriately, the set $\sols$ is nonempty iff the answer to the problem we are solving is true. 

Our goal is to count $|\sols_\targetW|$ modulo $2$ for all possible weights $\targetW$. In order to do this, we define the family of {\emph{objects}} $\objs_\targetW$ as the family of tuples 
$$((\vX,\vY,\Xmarkers,\Ymarkers),C_1,C_2,\ldots,C_{p_1},D_1,D_2,\ldots,D_{q_1}),$$
where
\begin{itemize}
\item $(\vX,\vY,\Xmarkers,\Ymarkers)\in\cand_\targetW$;
\item for $1\leq i\leq p_1$, $C_i=(V_{\eX_i,1},V_{\eX_i,2})$ is such that $(\eX_i,C_i)$ is a~consistently cut subgraph of $G$ and~$\xmarkers_i\subseteq V_{\eX_i,1}$;
\item for $1\leq j\leq q_1$, $D_j=(V_{\eY_j,1},V_{\eY_j,2})$ is such that $(\eY_j,D_j)$ is a~consistently cut subgraph of $G$ and~every edge from $\ymarkers_j$ has both endpoints in $V_{\eY_j,1}$.
\end{itemize}
Intuitively, an~object is a~candidate solution together with a~tuple of cuts consistent with the quantified sets, such that all the markers are on one side of the cut. Observe that in particular $V_{\eX_i,1} \cup V_{\eX_i,2}= \eX_i$ for all $1\leq i\leq p_1$ and~$V_{\eY_j,1} \cup V_{\eY_j,2}=V(\eY_j)$ for all $1\leq j\leq q_1$.
\vskip 0.3cm
\noindent{\bf The Count part.} We begin with the observation that the algorithm can count $|\objs_\targetW| \pmod{2}$ instead of $|\sols_\targetW| \pmod{2}$.
\begin{lemma}\label{lem:evencancel}
$|\objs_\targetW| \equiv |\sols_\targetW| \pmod{2}$ for all weights $\targetW$.
\end{lemma}
\begin{proof}
Let us consider a~candidate solution $Q=(\vX,\vY,\Xmarkers,\Ymarkers)$. For $1\leq i\leq p_1$ let us denote by $\compnomark(\eX_i,\xmarkers_i)$ the number of connected components of $G[\eX_i]$ not containing a~vertex from $\xmarkers_i$ (called further {\em{unmarked}}). Similarly, $\compnomark(\eY_j,\ymarkers_j)$ is the number of connected components of $G[\eY_j]$ not containing an~edge (arc) from $\ymarkers_j$ for $1\leq i\leq q_1$. Observe that there are exactly $\prod_{i=1}^{p_1}2^{\compnomark(\eX_i,\xmarkers_i)} \cdot \prod_{j=1}^{q_1}2^{\compnomark(\eY_j,\ymarkers_j)}$ objects associated with $Q$: for every set $\eX_i$ we have $2^{\compnomark(\eX_i,\xmarkers_i)}$ choices of including unmarked connected components of $G[\eX_i]$ to the sides of the cut $C_i$, and~the analogous holds for sets $\eY_j$. Therefore, 
$$|\objs_\targetW|=\sum_{(\vX,\vY,\Xmarkers,\Ymarkers)\in \cand_\targetW} \prod_{i=1}^{p_1}2^{\compnomark(\eX_i,\xmarkers_i)} \cdot \prod_{j=1}^{q_1}2^{\compnomark(\eY_j,\ymarkers_j)}.$$
Take both sides modulo $2$. Observe that the product under the sum is odd exactly for those candidate solutions, where there are no unmarked connected components of the quantified sets. Therefore, when considered modulo $2$, the summands are ones for these candidate solutions that are in fact solutions, and~zeroes otherwise. The claim follows.
\end{proof}

We now present a~dynamic programming routine that computes $|\objs_\targetW|$ modulo $2$ for all possible weights $\targetW$. We begin with adjusting the information stored in a~vertex to our needs. We will follow the notation from the proof of Theorem \ref{thm:deterministic}: $\psi_i=\dia^{S_i} \psi'_i$ are the subformulas of $\psi$ that begin with quantification, for $1\leq i\leq l$. Furthermore, $S_i=\alpha_{S_i}^{-1}(F_i)$ for homomorphisms $\alpha_{S_i}: \N \to M_i$ mapping $\N$ into finite monoids $M_i$, and~sets $F_i\subseteq M_i$. Let 
\begin{itemize}
\item $\Hh=\prod_{i=1}^l M_i$ be the history monoid; 
\item $\Pp=\{0,1\}^l$ be the set of possible predictions;
\item $\Xx=\{\zero, \oneone, \onetwo\}^{p_1}$ be the set of vectors indicating belonging to the sets $\eX_i$, including the side of the cut: $\zero$ means not belonging, $\oneone$ means belonging to $V_{\eX_i,1}$, $\onetwo$ means belonging to $V_{\eX_i,2}$;
\item $\Yy=\{\zero, \oneone, \onetwo\}^{q_1}$ be the set of vectors indicating existence of neighbouring edges from the sets $\eY_j$. $\zero$ on $j$-th coordinate means that there is no incident edge from $\eY_j$ introduced so far, $\oneone$ means that some incident edges have been introduced and~the vertex has been chosen to be in $V_{\eX_j,1}$, $\onetwo$ means the analogous but the vertex has been chosen to be in $V_{\eY_j,2}$.
\end{itemize}
Observe that in spite of syntactical similarities between sets $\Xx$ and~$\Yy$, their role is quite opposite. While information from $\Xx$ is being guessed in the introductory step of a~vertex, and~then is constant during considering possible extensions of a~partial solution, $\Yy$ acts more like history, remembering the types of so far introduced adjacent edges.

Let us define the set of possible information stored about a~vertex as $\Ii=\Hh\times \Pp\times \Xx\times \Yy$. We will refer to corresponding parts as to {\em{history}}, {\em{prediction}}, $\Xx${\em{-alignment}} and~$\Yy${\em{-alignment}}.

Let us fix a~bag $\bag{x}$. As~in the proof of Theorem \ref{thm:deterministic}, we define the set of partial objects $\objs_{x}(\targetW, \iterV, \iterE, \itercV, \itercE,s)$ for information evaluation $s\in \Ii^{\bag{x}}$, as the family of tuples $((\vX,\vY,\Xmarkers,\Ymarkers),C_1,C_2,\ldots,C_{p_1},D_1,D_2,\ldots,D_{q_1})$ such that following conditions are satisfied.
\begin{itemize}
\item The sum of weights of edges from $\subedges{x}$ and~vertices from $\subbags{x}\setminus\bag{x}$ is exactly $\targetW$.
\item For all $1\leq i\leq p_1$ the following holds:
\begin{itemize}
\item $(\eX_i, C_i)$ is a~consistently cut subgraph of $\subgraph{x}$;
\item $\xmarkers_i\subseteq V_{\eX_i,1}\setminus \bag{x}\subseteq  \eX_i\subseteq \subbags{x}$;
\item $|\eX_i|=\iterv_i$, $|\xmarkers_i|=\itercv_i$.
\end{itemize}
\item For all $1\leq j\leq q_1$ the following holds:
\begin{itemize}
\item $(\eY_j, D_j)$ is a~consistently cut subgraph of $\subgraph{x}$;
\item $\ymarkers_j\subseteq \eY_j\subseteq \subedges{x}$ and~every edge (arc) in $\ymarkers_j$ has both endpoints in $V_{\eY_j,1}$;
\item $|\eY_j|=\itere_j$, $|\ymarkers_j|=\iterce_j$.
\end{itemize}
\item The $\Xx,\Yy$-alignments from $s(v)$ for $v\in \bag{x}$ are consistent with sets $\vX,\vY$ and~cuts $C_1,C_2,\ldots,C_{p_1},D_1,D_2,\ldots,D_{q_1}$.
\item In every vertex of $v\in \subbags{x}\setminus \bag{x}$ the formula $\psi$ is true, when evaluated in $\subgraph{x}$ supplied with sets $\vFX, \vFY, \vX,\vY$. However, when trying to evaluate the boolean value of some formula $\psi_i$ in a~vertex from $\bag{x}$, we access the value in the prediction instead of actually evaluating the formula.
\item For every $v\in \bag{x}$, the number of neighbours of $v$ in $\subgraph{x}$ satisfying formula $\psi'_i$ ($1\leq i\leq l$), when accessed directly from $v$, maps in $\alpha_{S_i}$ to the $i$-th coordinate of the history from $s(v)$. Again, when evaluating formulas $\psi_i$ in vertices from $\bag{x}$, we access the value from the prediction instead of actually determining the outcome in $\subgraph{x}$.
\end{itemize}
Note that according to this definition, the vertex marker sets have to be disjoint with the bag and~summation of weights is carried out over vertices that are not in the bag. The algorithm will guess the alignment of a~vertex to marker sets and~update the weight during its forget step. If we chose otherwise, namely to perform updates during introduction, the problem with double counting would arise during the join.

Let us denote $A_x(\targetW,\iterV, \iterE, \itercV, \itercE,s)=|\objs_{x}(\targetW, \iterV, \iterE, \itercV, \itercE,s)| \pmod{2}$. From now on, all the computations over the values $A_x$ will be carried out in $\Z_2$. Observe that we need to compute $\sum_{\itercV\text{: }\itercv_i\leq \eitercv_i}\sum_{\itercE\text{: }\iterce_i\leq \eiterce_i} A_r(\targetW,\eiterV,\eiterE,\itercV,\itercE,\emptyset)$ for all possible $\targetW$. Thus, it suffices to show a~dynamic programming routine that will compute all the values of $A_x$ for possible arguments in a~bottom-up fashion. We now present the steps that have to be carried out during computation for every type of a~bag. We follow convention that all values $A_x$ with improper arguments, for example having negative coordinates, are defined to be zeroes. For a~condition $c$, by $[c]$ we denote $1$ if $c$ is true, and~$0$ otherwise. Moreover, for a~function $s$ by $s[v \to \alpha]$ we denote the function $s \setminus \{(v,s(v))\} \cup \{(v,\alpha)\}$. Note that this definition is correct even when $s$ is not defined on $v$. Also, we treat vectors over $\{\zero,\oneone,\onetwo\}$ also as vectors over $\{0,1\}$ by mapping $\zero\to 0$ and~$\oneone, \onetwo \to 1$. 
\vskip 0.2cm
\noindent\textbf{Leaf bag $x$}:
	$$ A_x(0,\overline{0},\overline{0},\overline{0},\overline{0},\emptyset) = 1 $$
	$\overline{0}$ denotes vector of zeroes of appropriate length. All other values of $A_x(\targetW,\iterV,\iterE,\itercV,\itercE,\emptyset)$ are zeroes.
\vskip 0.2cm
\noindent\textbf{Introduce vertex $v$ bag $x$ with child $y$}:
	$$ A_x(\targetW,\iterV,\iterE,\itercV,\itercE,s[v \to (h,\pi,b,e)]) = [h=e_\Hh][e=\overline{\zero}]A_y(\targetW,\iterV-b,\iterE,\itercV,\itercE,s) $$
	This step is almost the same as in the algorithm from the Theorem \ref{thm:deterministic}. The new vertex has no neighbours so far, therefore its history must be void, which is indicated by checkng whether the assigned history is equal to $e_\Hh$, the identity of $\Hh$. For the same reason, its side of the cut for any set $\eY_j$ is not decided yet, hence the second check. However, prediction and~the $\Xx$-alignment can be arbitrary. Note that the new vertex does not contribute to the weight of the partial object and~does not belong to any marker sets.
\vskip 0.2cm
\noindent\textbf{Introduce edge (arc) $uv$ bag $x$ with child $y$}:

	\noindent Let $b(u),e(u),b(v),e(v)$ be the $\Xx$- and~$\Yy$-alignments in $s(u),s(v)$ respectively.
	\begin{eqnarray*}
	 A_x(\targetW,\iterV,\iterE,\itercV,\itercE,s) & = & [\forall_i (b(u)_i=\zero \vee b(v)_i=\zero \vee b(u)_i=b(v)_i)] \\
& & \sum_{d\in\{0,1\}^{q_1}} \ \ [\forall_j ((d_j=1) \Rightarrow (e(u)_j=e(v)_j\neq \zero))] \\
& & \sum_{m\in\{0,1\}^{q_1}} \ \ [\forall_j((d_j=0) \Rightarrow (m_j=0))\wedge \\
& & \qquad \qquad((m_j=1) \Rightarrow (e(u)_j=e(v)_j=\oneone))]\\ 
& & \sum_{s'\in S'} \ \ A_y(\targetW-\omega(uv,d+m),\iterV,\iterE-d,\itercV,\itercE-m,s')
	\end{eqnarray*}
	Before we start any summations, we need to ensure that the new edge is consistent with the cuts $C_i$, otherwise the whole outcome is zero. The first two summations correspond to all possible ways of choosing the alignment of the newly introduced edge to sets $\eY_j$ and~marker sets $\ymarkers_j$. Again, having fixed these alignments we have to ensure that they are consistent with the cuts $D_j$. As~the edge already contributes both to the cardinalities of marker sets and~the weight of the partial object, we need to access the precomputed values with updated weight and~cardinalities of sets $\vY$, $\Ymarkers$. In the third summation we sum over all information evaluations $s'$ such that $s'$ differs from $s$ only on histories and~$\Yy$-alignments of vertices $u,v$. As~in the proof of Theorem \ref{thm:deterministic}, histories in $s(u),s(v)$ have to be histories in $s'(u),s'(v)$ but updated with respect to the introduced edge by possibly adding an~image of one on a~coordinate, whenever a~formula $\psi_i'$ is true in the neighbour when accessed directly from the considered vertex. This can be resolved in constant time knowing vector $d$ and~the predictions and~$\Xx$--alignments in $u,v$. In addition, we need to ensure that the $\Yy$-alignment is properly updated: in both vertices $u,v$, for every index $i$, the $i$-th coordinate of the $\Xx$-alignment has to be at least the same in $s$ as in $s'$ (it may change from $\zero$ to $\oneone$ or~$\onetwo$, or~stay the same). Similarly as in the proof of Theorem \ref{thm:deterministic}, the number of contributing information evaluations $s'$ is constant and~the algorithm can enumerate them in constant time. Thus, the computation of a~single value can be performed in constant time.
\vskip 0.2cm
\noindent\textbf{Forget vertex $v$ bag $x$ with child $y$}:

	\noindent Let $b(v)$ denote the $\Xx$-alignment in $s(v)$.
	\begin{eqnarray*}
	 A_x(\targetW,\iterV,\iterE,\itercV,\itercE,s) & = & \sum_{m\in \{0,1\}^{p_1}} \ \ [\forall_i (m_i=1)\Rightarrow (b(v)_i=\oneone)] \sum_{(h,\pi,b,e)\in \Gg} \\ & & A_y(\targetW-\omega(v,b(v)+m),\iterV,\iterE,\itercV-m,\itercE,s[v\to (h,\pi,b,e)])
	\end{eqnarray*}
	The first summation corresponds to possible choices of vector $m$ indicating belonging of $v$ to the marker sets $\Xmarkers$. If $v$ is to be contained in $\xmarkers_i$, then it has to be contained in $V_{\eX_i,1}$, so the $i$-th coordinate of vector $b$ has to be $\oneone$. For a~particular vector $m$ already satisfying this condition, vector $b(v)+m$ (over $\{0,1,2\}$) exactly indicates the belonging of $v$ to $\vX$ and~$\Xmarkers$ in the sense of the definition of weight function $\omega$. Thus $\omega(v,b(v)+m)$ is the precise weight of vertex $v$ in this partial object and~can be used to access precomputed value with updated weight. The second summation is the same as in the corresponding step of the algorithm from Theorem \ref{thm:deterministic}. We sum over all possible information that could be stored in the vertex we forget, i.e., having the history consistent with the prediction and~making the formula $\psi$ satisfied. Similarly to the proof of Theorem \ref{thm:deterministic}, the computation of a~single value can be carried out in constant time.
\vskip 0.2cm
\noindent\textbf{Join bag $x$ with children $y,z$}:
	\begin{eqnarray*} 
	A_x(\targetW,\iterV,\iterE,\itercV,\itercE,s) & = & \sum_{\targetW'+\targetW''=\targetW} \ \ \sum_{\iterV'+\iterV''=\iterV+\xi(\bag{x})} \ \ \sum_{\iterE'+\iterE''=\iterE}\ \ \sum_{\itercV'+\itercV''=\itercV} \ \ \sum_{\itercE'+\itercE''=\itercE}\ \ \sum_{s'+s''=s} \\ & & A_y(\targetW',\iterV',\iterE',\itercV',\itercE',s')A_z(\targetW'',\iterV'',\iterE'',\itercV'',\itercE'',s'')
	\end{eqnarray*}
	The step is a~generalization of the corresponding from the proof of Theorem \ref{thm:deterministic}. Here, $s'+s''=s$ denotes that for all the vertices $v\in \bag{x}$: \begin{itemize}
\item the predictions in $s(v),s'(v),s''(v)$ are the same;
\item the $\Xx$-alignments in $s(v),s'(v),s''(v)$ are the same;
\item the histories in $s'(v)$ and~$s''(v)$ sum up to the history in $s(v)$ (in the history monoid);
\item the $\Yy$-alignments in $s'(v)$ and~$s''(v)$ sum up to $\Yy$-alignment in $s(v)$. By this, we mean that $\zero+a=a$ for all $a\in \{\zero, \oneone,\onetwo\}$, $\oneone+\oneone=\oneone$, $\onetwo+\onetwo=\onetwo$, however addition $\oneone+\onetwo$ cannot be carried out and~such a~pair is forbidden to occur on any coordinate of added vectors.
\end{itemize}
The first summation corresponds to splitting the weight among two partial solutions, the next two correspond to splitting the cardinalities of sets $\vX$, $\vY$, the next two correspond to splitting the numbers of so far used markers and~the last summation corresponds to splitting the information evaluations. As~the marker sets are disjoint with the bags, weights of the partial solution are not summed over the bag and~$\subedges{x}$ is a~disjoint sum of $\subedges{y}$ and~$\subedges{z}$, the problem with double counting can possibly occur only in the second sum. It can be however solved by adding to the right side of the equation the vector $\xi(\bag{x})$ --- the sum over the bag $\bag{x}$ of $\Xx$-alignments in $s$. Similarly as in the proof of Theorem \ref{thm:deterministic}, for every pair of information evaluations $(s',s'')$ the algorithm can determine the (at most one) information evaluation $s$ it contributes to. Then, for every information evaluation $s$ we consider only contributing pairs, thus considering every pair only once. Therefore, the computation of the whole step can be performed in time $|\Ii|^{2t}|G|^{O(1)}$.

The computation of a~single value in leaf, introduce, introduce edge and~forget steps takes constant time, while the whole join step can be performed in $|\Ii|^{2t}|G|^{O(1)}$ time. As~there are $|\Ii|^t|G|^{O(1)}$ values to be computed at each step and~the number of steps is polynomial, the whole dynamic programming routine runs in $|\Ii|^{2t}|G|^{O(1)}$ time. 

The whole algorithm works as follows. Firstly, choose randomly the weight function, each value independently with uniform distribution. Then, for every possible weight $\targetW$ compute $|\sols_\targetW|$ modulo $2$ using described dynamic programming routine and~Lemma \ref{lem:evencancel}. If at least one of the computed values is $1$, answer YES, otherwise answer NO.

In order to prove soundness of the described algorithm, observe that if at least one $|\sols_\targetW|$ is odd then it is a~sufficient proof of existence of at least one solution. Therefore, the algorithm can safely answer YES without risking a~false positive. On the other hand, the Isolation Lemma assures that in case of existence of solutions, i.e. the set $\sols$ being nonempty, with probability at least $\frac{1}{2}$ there exists a~unique solution with minimal weight $\targetW_0$. As~$1$ is odd, $|\sols_{\targetW_0}|$ is odd as well and~the algorithm will answer YES.
\end{proof}

\section{\laglc formulas for problems considered in \cite{my}}\label{sec:formulae}

We present logical formulas of \laglc for problems proven to be tractable in single exponential time when parameterized by treewidth by Cygan et al., when introducing the Cut\&Count technique~\cite{my}. All of them are of quantification rank at most $1$, which explains why Cygan et al. did not need to use the prediction technique in their proofs. The exact problem definitions can be found in~\cite{my}.

The formulas do not use fixed sets, unless it is explicitely stated. The vectors of parameters always consist of one parameter $k$. 

\vskip 0.5cm 

\hrule \vskip 0.2cm
\steinertree

($T$, the terminals, is a~fixed set of vertices)
\begin{align*}
\exists_{X\subseteq V} (|\compnomark(X)|\leq 1 \wedge |X|\leq k+|T|) \wedge \forall_v G,T,X,v\models (T\Rightarrow X)
\end{align*}
\hrule \vskip 0.2cm

\fvs
\begin{align*}
\exists_{X\subseteq V}\exists_{Z\subseteq V}\exists_{Y\subseteq E} (|\compnomark(Y)|+|Y|+|Z|+|X|\leq|V| \wedge |X|\leq k) \wedge \\ \forall_v G,X,Z,Y,v\models \left[(Z\Leftrightarrow (\neg X\wedge \rect X))\wedge (X\Rightarrow \rect \neg Y) \wedge (\neg X \Rightarrow \rect (\neg X \Rightarrow Y))\right]
\end{align*}
\hrule \vskip 0.2cm

\cvertexcover
\begin{align*}
\exists_{X\subseteq V} (|\compnomark(X)|\leq 1 \wedge |X|\leq k) \wedge \forall_v G,X,v\models (\neg X \Rightarrow \rect X)
\end{align*}
\hrule \vskip 0.2cm

\cdomset
\begin{align*}
\exists_{X\subseteq V} (|\compnomark(X)|\leq 1 \wedge |X|\leq k) \wedge \forall_v G,X,v\models (\neg X \Rightarrow \dia X)
\end{align*}
\hrule \vskip 0.2cm

\cfvs
\begin{align*}
\exists_{X\subseteq V} \exists_{Z\subseteq V}\exists_{Y\subseteq E} (|\compnomark(Y)|+|Y|+|Z|+|X|\leq|V|\wedge |\compnomark(X)|\leq 1 \wedge |X|\leq k ) \wedge \\ \forall_v G,X,Z,Y,v\models \left[(Z\Leftrightarrow (\neg X\wedge \rect X))\wedge (X\Rightarrow \rect \neg Y) \wedge (\neg X \Rightarrow \rect (\neg X \Rightarrow Y))\right]
\end{align*}
\hrule \vskip 0.2cm

\coct
\begin{align*}
\exists_{X\subseteq V} \exists_{L\subseteq V} \exists_{R\subseteq V} (|\compnomark(X)|\leq 1 \wedge |X|\leq k) \wedge \forall_v G,X,L,R,v\models \\
(L\vee R\vee X) \wedge \neg (L\wedge R) \wedge \neg (R\wedge X) \wedge \neg (X\wedge L) \wedge\\
(L\Rightarrow \rect (R\vee X)) \wedge (R\Rightarrow \rect (L\vee X))
\end{align*}
\hrule \vskip 0.2cm

Undirected \mincyclecovername
\begin{align*}
\exists_{Y\subseteq E} (|\compnomark(Y)|\leq k) \wedge \forall_v G,Y,v\models \dia^{\{2\}} Y
\end{align*}
\hrule \vskip 0.2cm

Directed \mincyclecovername
\begin{align*}
\exists_{Y\subseteq E} (|\compnomark(Y)|\leq k) \wedge \forall_v G,Y,v\models \left[(\dia^{\{1\}} (Y\wedge \uparrow)) \wedge (\dia^{\{1\}} (Y\wedge \downarrow))\right]
\end{align*}
\hrule \vskip 0.2cm

Undirected \longestpath
\begin{align*}
\exists_{A\subseteq V}\exists_{Y\subseteq E} (|\compnomark(Y)|\leq 1 \wedge |A|=2 \wedge |Y|\geq k) \wedge \forall_v G,A,Y,v\models \\
\left[(A \Rightarrow \dia^{\{1\}} Y)\wedge (\neg A~\Rightarrow \dia^{\{0,2\}} Y)\right]
\end{align*}
\hrule \vskip 0.2cm

Directed \longestpath
\begin{align*}
\exists_{A\subseteq V}\exists_{B\subseteq V}\exists_{Y\subseteq E} (|\compnomark(Y)|\leq 1 \wedge |A|=1 \wedge |B|=1 \wedge |Y|\geq k) \wedge \forall_v G,A,B,Y,v\models \\ 
\left(A \Rightarrow \left[\neg B \wedge \dia^{\{1\}} Y \wedge \dia^{\{1\}} (Y\wedge \downarrow))\right]\right)\wedge \\
\left(B \Rightarrow \left[\neg A~\wedge \dia^{\{1\}} Y \wedge \dia^{\{1\}} (Y\wedge \uparrow))\right]\right)\wedge \\
\left((\neg A\wedge \neg B) \Rightarrow \left[(\neg \dia Y) \vee ((\dia^{\{1\}} (Y\wedge \downarrow)) \wedge (\dia^{\{1\}}(Y\wedge \uparrow)))\right]\right)
\end{align*}
\hrule \vskip 0.2cm

Undirected \longestcycle
\begin{align*}
\exists_{Y\subseteq E} (|\compnomark(Y)|\leq 1 \wedge |Y|\geq k) \wedge \forall_v G,Y,v\models \dia^{\{0,2\}} Y
\end{align*}
\hrule \vskip 0.2cm

Directed \longestcycle
\begin{align*}
\exists_{Y\subseteq E} (|\compnomark(Y)|\leq 1 \wedge |Y|\geq k) \wedge \forall_v G,Y,v\models \left[(\neg \dia Y) \vee ((\dia^{\{1\}} (Y\wedge \downarrow)) \wedge (\dia^{\{1\}}(Y\wedge \uparrow)))\right]
\end{align*}
\hrule \vskip 0.2cm

\exactleaf
\begin{align*}
\exists_{L\subseteq V}\exists_{T\subseteq E} (|\compnomark(T)|\leq 1 \wedge |L|=k \wedge |T|=|V|-1) \wedge \forall_v G,L,T,v\models (\dia T) \wedge (L\Leftrightarrow \dia^{\{1\}} T)
\end{align*}
\hrule \vskip 0.2cm

\exactoutbranching 

($R$, the singleton of the root, is a~fixed set)
\begin{align*}
\exists_{L\subseteq V}\exists_{T\subseteq E} (|\compnomark(T)|\leq 1 \wedge |L|=k \wedge |T|=|V|-1 \wedge |R|=1) \wedge \forall_v G,R,L,T,v\models \\
\left[(\dia T) \wedge (R\Rightarrow \neg \dia (T\wedge \uparrow)) \wedge (\neg R\Rightarrow \dia^{\{1\}} (T\wedge \uparrow)) \wedge (L\Leftrightarrow \neg \dia (T\wedge \downarrow))\right]
\end{align*}
\hrule \vskip 0.2cm

\maxspantree
\begin{align*}
\exists_{F\subseteq V}\exists_{T\subseteq E} (|\compnomark(T)|\leq 1 \wedge |F|\geq k \wedge |T|=|V|-1) \wedge \forall_v G,F,T,v\models (\dia T) \wedge (F \Leftrightarrow \rect T)
\end{align*}
\hrule \vskip 0.2cm

\gmtsp 

($2\mathbb{N}$ denotes the set of even nonnegative integers)
\begin{align*}
\exists_{Y\subseteq E}\exists_{Y_1\subseteq E}\exists_{Y_2\subseteq E} (|\compnomark(Y)|\leq 1 \wedge |Y_1|+2|Y_2|\leq k) \wedge \forall_v G,Y,Y_1,Y_2,v\models \\
\left[(\rect (Y\Leftrightarrow (Y_1\vee Y_2))) \wedge (\rect (\neg Y_1\vee \neg Y_2)) \wedge (\dia Y) \wedge (\dia^{2\mathbb{N}} Y_1)\right]
\end{align*}

\hrule

\section{Correctness of the reduction from the proof of Theorem~\ref{thm:negative}}\label{sec:negative-details}

\subsubsection*{Soundness.} Let $G$ be the graph obtained in the construction. Denote $P=A\cup B$. We prove the soundness of the construction in two steps, as was described in Section~\ref{sec:negative}. We also follow convention introduced there. 

\begin{lemma}\label{lem:soundness1}
If $\ff$ is satisfiable, then $(G,k)$ is a~YES instance of \ltlvertexdeletion.
\end{lemma}
\begin{proof}
We need to show that $G$ contains a~set $X$ of $n+2m$ vertices that hits all the cycles of length at most $l$. Let $\phi$ be an~assignment satisfying $\ff$. For every variable $x$, take into $X$ the vertex $t_x$ if $\phi(x)=TRUE$, and~$t_{\neg x}$ otherwise. For every clause $S=r_1\vee r_2\vee r_3$ let $r_i$ be any literal that satisfies $S$. Take into $X$ two vertices $r_j$, where $j\neq i$. Thus $|X|=n+2m$. We now verify that $G\setminus X$ contains no cycle of length at most $l$.

Let $C$ be any cycle in $G$. We need to prove that either $C$ contains a~vertex from $X$ or~is of length greater than $l$. Observe that the parts of gadgets not contained in $P$ are pairwise independent. Therefore, we can distinguish three cases: 
\begin{itemize}
\item $C$ is fully contained in one gadget;
\item $C$ is not contained in one gadget and~passes through exactly two vertices from $P$;
\item $C$ is not contained in one gadget and~passes through three or~more vertices from $P$.
\end{itemize}

Regarding the first case, observe that after deleting $X$ every gadget becomes a~forest, so in this case $C$ has to contain a~vertex from $X$. Note that this is true also when some of vertices $u,v,u',v'$ (in case of the vertex gadget) or~$u_i,v_i$ (in case of the clause gadget) coincide.

Regarding the latter cases, observe that two vertices $u,v$ from $P$ in $G$ are in distance $\alpha$ if $(u,v)=\psi(r)$ for some literal $r$, and~are in distance at least $\alpha+2$ otherwise. If $C$ passes through two vertices from $P$ that are in distance at least $\alpha+2$, then its length is at least $2\alpha+4>l$. This immediately resolves the third case: if $C$ passes through at least three vertices from $P$, then there is a~pair of them contained either both in $A$ or~both in $B$, thus not contained in the image of $\psi$. As~a~result, in the third case the length of $C$ is greater than $l$.

We are left with the second case. Moreover, we can focus only on the subcase, when the two vertices from $P$ that $C$ passes through are such $u,v$ that $(u,v)=\psi(r)$ for some literal $r$, equal to $x$ or~$\neg x$. Observe that paths connecting $u$ and~$v$ in $G\setminus X$ not passing through other vertices from $P$, can only be paths built while constructing gadgets $Q_x$ or~$C_{S,r}$ for clauses $S$ containing $r$, of length $\alpha$ and~$\beta$ respectively. As~$2\beta >l$, $C$ could possibly not pass through vertices from $X$ and~have length at most $l$, if it consisted of the path from $Q_x$ and~a~path from $C_{S,r}$ for some $S$. If $\phi(x)$ is such that $r$ is true, then $t_{r}\in X$ and~$C$ contains a~vertex from $X$. Otherwise, all the clauses containing $r$ had to be satisfied by some other literal, so $s_{S,r}\in X$ for every $S$ containing $r$. Thus, also in this situation $C$ contains a~vertex from $X$.
\end{proof}

\begin{lemma}\label{lem:soundness2}
If $(G,k)$ is a~YES instance of \lvertexdeletion, then $\ff$ is satisfiable.
\end{lemma}
\begin{proof}
Let $X$ be the set of at most $n+2m$ vertices such that $G\setminus X$ contains no cycles of length $l$. Observe that $X$ has to include at least one vertex from each cycle of length $l$ spanned between vertices $t_x$ and~$t_{\neg x}$ for every variable $x$, and~at least two vertices from each subgraph induced by a~triple of cycles of length $l$ spanned between $s_{S,r_1}$, $s_{S,r_2}$, $s_{S,r_3}$ for every clause $S=r_1\vee r_2\vee r_3$. All the mentioned subgraphs are pairwise disjoint, so $X$ has to contain exactly one vertex from each cycle spanned between $t_x$ and~$t_{\neg x}$ and~exactly two vertices from each subgraph induced by a~triple of cycles spanned between $s_{S,r_1}$, $s_{S,r_2}$, $s_{S,r_3}$. Observe that we can assume that the solution does not contain any inner vertex of these cycles, i.e., of degree $2$, because a~choice of such a~vertex can always be substituted with a~choice of $t_x$, $t_{\neg x}$ or~$s_{S,r_i}$ for some $i$ (depending whether we are considering a~variable or~a~clause gadget). Therefore, for each variable $x$, the set $X$ contains exactly one vertex from the set $\{t_x,t_{\neg x}\}$ and~for each clause $S=r_1\vee r_2\vee r_3$, $X$ contains exactly two vertices from the set $\{s_{S,r_1},s_{S,r_2},s_{S,r_3}\}$. Consider an~assignment $\phi$ such that $\phi(x)=TRUE$ if $t_x\in X$ and~$\phi(x)=FALSE$ if $t_{\neg x}\in X$. We claim that $\phi$ satisfies $\ff$.

Consider a~clause $S=r_1\vee r_2\vee r_3$. Let $i$ be such an~index that $s_{S,r_i}\notin X$. Consider a~cycle of length $l$ formed by two paths connecting vertices from $\psi(r_i)$: one from the gadget $C_S$ of length $\beta$ and~one from the gadget $Q_x$ of length $\alpha$, where $r=x$ or~$r=\neg x$. As~$X$ hits this cycle, then $t_{r_i}\in X$, so $r_i$ satisfies $S$. As~$S$ was an~arbitrary clause, this concludes the proof.
\end{proof}
\vskip 0.5cm
\subsubsection*{The bound on pathwidth.} 

\begin{lemma}\label{lem:pathwidth}
$\pw(G)=O(\sqrt{n})$ and~a~decomposition of such width can be computed in polynomial time.
\end{lemma}
\begin{proof}
As was already mentioned in the proof of Lemma \ref{lem:soundness1}, the parts of gadgets not contained in $P$ are pairwise independent. Moreover, the gadgets are of constant size. Therefore, we can create a~path decomposition of width $O(\sqrt{n})$ in the following manner. We construct $n+m$ bags, one for each gadget. The bag contains the whole set $P$ and~the whole gadget, thus having size $O(\sqrt{n})$. We arrange the bags into a~path in any order.
\end{proof}

\end{document}